\documentclass{sig-alternate}

\usepackage{epsfig}
\usepackage{times,graphicx,xspace,comment,subfigure,algpseudocode,multirow,url,amsmath}

\usepackage[ruled,vlined,oldcommands]{algorithm2e}

\usepackage{booktabs}
\usepackage{tabularx}

\begin{document}

%

\newtheorem{thm}{Theorem}
\newtheorem{theorem}{Theorem}
\newtheorem{cor}[thm]{Corollary}
\newtheorem{lemma}[thm]{Lemma}
\newtheorem{conjecture}[thm]{Conjecture}
\newtheorem{prop}[thm]{Proposition}
\newtheorem{claim}[thm]{Claim}
\newtheorem{fact}[thm]{Fact}
\newcommand{\assign}{\leftarrow}
\newcommand{\be}{\begin{equation}}
\newcommand{\ee}{\end{equation}}
\newcommand{\bea}{\begin{eqnarray}}
\newcommand{\eea}{\end{eqnarray}}
\newcommand{\Mmax}{M_{\textrm{max}}}
\newcommand{\eol}{\end{enumerate}\setlength{\itemsep}{-\parsep}}
\newcommand{\eg}{\emph{e.g.}\xspace}
\newcommand{\ie}{\emph{i.e.}\xspace}
\newcommand{\etal}{\emph{et al.}\xspace}

\let\eps\varepsilon\relax
\def\up#1{^{(#1)}}
\def\down#1{_{(#1)}}
\let\hat\widehat
\def\Rl{{\mathbb R}}
\def\Nl{{\mathbb N}}
\def\Ir{{\mathbb Z}}
\def\E{{\mathbb E}}
\def\Pr{{\mathbb P}}
\def\cP{{\cal P}}
\def\cF{{\cal F}}
\def\cH{{\cal H}}
\def\cK{{\cal K}}
\def\ol#1{{\overline{#1}}}
\def\var{\mathop{\mathrm{Var}}}
\def\cov{\mathop{\mathrm{Cov}}}
\def\nf{\hat n_{{\textsc f}}}
\def\kf{\hat K_{{\textsc f}}}
\def\kfone{\hat K_{{\textsc i}}}
\def\nfp{\hat n_{\textsc{fp}}}

\title{Graph Sample and Hold: A Framework for Big-Graph Analytics}

\numberofauthors{4}
\author{
\alignauthor
Nesreen K. Ahmed\\
       \affaddr{Purdue University}\\
       \email{nkahmed@cs.purdue.edu}
\and
Nick Duffield\\
       \affaddr{Rutgers University}\\
       \email{nick.duffield@rutgers.edu}
\and
Jennifer Neville\\
       \affaddr{Purdue University}\\
       \email{neville@cs.purdue.edu}
\and
Ramana Kompella\\
       \affaddr{Purdue University} \\
       \email{kompella@cs.purdue.edu}
}

\maketitle




\begin{abstract}
Sampling is a standard approach in big-graph analytics; the goal is to efficiently estimate the graph properties by consulting a sample of the whole population.  A perfect sample is assumed to mirror every property of the whole population. Unfortunately, such a perfect sample is hard to collect in complex populations such as graphs (e.g. web graphs, social networks etc), where an underlying network connects the units of the population. Therefore, a good sample will be representative in the sense that graph properties of interest can be estimated with a known degree of accuracy.

While previous work focused particularly on sampling schemes used to estimate certain graph properties (e.g. triangle count), much less is known for the case when we need to estimate various graph properties with the same sampling scheme. In this paper, we propose a generic stream sampling framework for big-graph analytics, called Graph Sample and Hold (gSH).  To begin, the proposed framework samples from massive graphs sequentially in a single pass, one edge at a time, while maintaining a small state. We then show how to produce unbiased estimators for various graph properties from the sample. Given that the graph analysis algorithms will run on a sample instead of the whole population, the runtime complexity of these algorithm is kept under control.  Moreover, given that the estimators of graph properties are unbiased, the approximation error is kept under control. Finally, we show the performance of the proposed framework (gSH) on various types of graphs, such as social graphs, among others.
\end{abstract}


\section{Introduction}
\label{sec-intro}

\subsection{Motivation}
We live in a vastly connected world. A large percentage of world's population
routinely use online applications (e.g., Facebook and instant messaging) that allow them to interact with their friends, family, colleagues and
anybody else that they wish to. Analyzing various properties of these
interconnection networks is a key aspect in managing these applications; for
example, uncovering interesting dynamics often prove crucial for either enabling
new services or making existing ones better. Since these interconnection
networks are often modeled as graphs, and these networks are huge in
practice (e.g., Facebook has more than a billion nodes), efficient {\em
big-graph analytics} has recently become extremely important.

One key stumbling block for enabling big graph analytics is the limitation in
computational resources. Despite advances in distributed and parallel processing
frameworks such as MapReduce for graph analytics and the appearance of infinite
resources in the cloud, running brute-force graph analytics is either too
costly, too slow, or too inefficient in many practical situations. Further,
finding an `approximate' answer is usually sufficient for many types
of analyses; the extra cost and time in finding the exact answer is often not
worth the extra accuracy. {\em Sampling} therefore provides an
attractive approach to quickly and efficiently finding an approximate
answer to a query, or more generally, any analysis objective. 

Many interesting graphs in the online world naturally evolve over time, as new
nodes join or new edges are added to the network. A natural representation of
such graphs is in the form of a stream of edges, as some prior work noted
\cite{ahmed2013network}.  Clearly, in such a streaming graph model, sampling algorithms
that process the data in one-pass are more efficient than those that process
data in an arbitrary order.  Even for static graphs, the streaming model is
still applicable, with a one-pass algorithm for processing arbitrary queries
over this graph typically more efficient than those that involve arbitrary
traversals through the graph. 

\subsection{Sampling, Estimation, Accuracy}
In this paper, we propose a new sampling framework for big-graph
analytics, called Graph Sample and Hold (gSH). gSH essentially
maintains a small amount of state and passes through all edges in the
graph in a streaming fashion. The sampling probability of an arriving
edge can in general be a function of the stored state, such as the adjacency
properties of the arriving edge with those already sampled. (This can
be seen as an analog of the manner in which standard Sample and Hold
\cite{EV:02} samples packets with a probability depending on whether
their key matches one already sampled). Since the
algorithm involves processing only a sample of edges (and thus,
nodes), it keeps run time complexity under check.  


gSH provides a generic framework for unbiased estimation of the counts
of arbitrary subgraphs. This uses the
Horvitz-Thompson construction \cite{HT52} in which the count of any
sampled object is weighted by dividing by its sampling probability. In
gSH this is realized by maintaining along with each sampled edge, the
sampling probability that was in force when it was sampled. The counts
of subgraphs of sampled edges are then weighted
according to the product of the selection probabilities of their
constituent edges. Since the edge sampling probabilities are
determined conditionally 
with respect to the prior sampling outcomes, this product
reflects the dependence structure of edge selection.

The sampling framework also provide the means to compute the accuracy
of estimates, since the unbiased estimator of the variance of the
count estimator can be computed from the sampling probabilities of
selected edges alone. More generally, the covariance between the count
estimators of any pair of subgraphs can be estimated in the same manner.

The framework itself is quite generic. By varying the dependence of
sampling probabilities on previous history, one can tune the 
estimation various properties of the original graph efficiently
with arbitrary degrees of accuracy. For example, simple uniform sampling of
edges at random may naturally lead to selecting a large number of higher-degree
nodes since higher-degree nodes appear in more number of edges. For each of
these sampled nodes, we can choose the holding function to simply track the size
of the degree for these specific nodes, of course accounting for the loss of the
count before the node has been sampled in an unbiased manner. Similarly, by
carefully designing the sampling function, we can obtain a uniformly random
sample of nodes (similar to the classic node sampling), for whom we can choose
to hold an accurate count of number of triangles each of these nodes is part of.

\subsection{Applications of the gSH Framework}
In this paper, we demonstrate applications of the gSH framework in two
directions. Firstly, we formulate a parameterized family gSH(p,q) of
gSH sampling schemes, in which an arriving edge with no adjacencies
with previously sampled edges is selected with probability $p$;
otherwise it is sampled with probability $q$. Secondly, we consider
four specific quantities of interest to estimate within the framework.
These are counts of
links, triangles, connected paths of length two, and the derived
global clustering coefficient. We also provide an unbiased estimator
of node counts based on edge sampling. 
Note that we do not claim that these lists of examples are by any
means exhaustive or that the framework can accommodate arbitrary queries
efficiently. 

\subsection{Contributions and Outline}

In Section~\ref{sec-framework}, we describe the general framework for
graph sampling, and show how it can be used to
provide unbiased estimates of the counts of arbitrary selections 
of subgraphs. We also show how unbiased estimates of the variance of
these estimators can be efficiently computed within the same
framework. In Section~\ref{sec-3}, we show how counts of specific types
of subgraph (links, triangles, paths of length 2) and the global
clustering coefficient can be estimated in this framework. In
Section~\ref{sec-4}, we describe the specific gSH(p,q) graph Sample and
Hold algorithms, and illustrate the application of gSH(p,1) on a
simple graph. In Section~\ref{sec-experiments}, we describe a set of
evaluations based on a number of real network topologies. We apply the
estimators described in Section~\ref{sec-4} to the counts described in
Section~\ref{sec-3}, and compare empirical confidence intervals with
those estimated directly from the samples. We also compare accuracy
with prior work. We discuss the general relation of our work to
existing literature in Section~\ref{sec-related} and conclude in Section~\ref{sec-conclusion}.

\subsection{Relation to Sample and Hold}
gSH for big-graph analytics bears some resemblance to the classic
Sample and Hold (SH) approach \cite{EV:02}, versions of which also
appeared as Counting Samples of Gibbons and Matias\cite{GM:sigmod98},
and  were used for attack detection by Smitha, Kim and Reddy
\cite{Reddy2001}. 
In SH, packets carry a key that identifies the flow to which
they belong. A router maintains a cache of information concerning the
flows of packets that traverse it.
If the key of an arriving packet matches a key on
which information is currently maintained in the router, the information
for that key (such as packet and byte counts and timing information)
is updated accordingly. Otherwise the packet is sampled with some
probability $p$. If selected, a new entry is instantiated in the
cache for that key. SH is more likely to sample longer flows.
Thus, SH provides an efficient way to store
information concerning the disposition of packet across the small
proportion of flows that carry a large proportion of all network
packets.

gSH can be viewed as an analog of SH in which the equivalence relation
of packets according to their keys is replaced by adjacency
relation between links. But this generalization brings
many differences as well.
In particular, many graph properties involve transitive properties (e.g.,
triangles) that are relatively uninteresting in network measurements (and hence,
under explored).  For many of these properties, it is important to realize that
the accuracy of the analytics depends on the ordering of edges to some extent,
which was not the case for the vast majority of network measurement problems
considered in the literature.

\section{Framework for Graph Sampling}
\label{sec-framework}

\subsection{Graph Stream Model}
Let $G=(V,K)$ be a graph. 
We call two edges $k,k'\in K'$ are adjacent, $k\sim k'$, if they join at some node. Specifically:
\begin{itemize}
\item \textsl{Directed adjacency:} $k=(k_1,k_2)\sim k'=(k'_1,k'_2)$ iff $k_2=k'_1$ or $k_1=k'_2$. Note that $\sim$ is not symmetric in this case.
\item \textsl{Undirected adjacency:} $k=(k_1,k_2)\sim k'=(k'_1,k'_2)$ iff 
$k\cap k'\ne \emptyset$.  Note that $\sim$ is symmetric in this case.
\end{itemize}
Without loss of generality we assume edges are unique; otherwise 
distinguishing labels that are ignored by $\sim$\ can be appended.

The edges in $K$ arrive in an order $k:[|K|]\to K$. 
For $k,k'\in K$, we write $k\prec k'$
if $k$ appears earlier than $k'$ in arrival order. For $i\le |K|$,
$K_i=\{k\in K: k\preceq k_i\}$ comprises the first $i$ arrivals. 

\subsection{Edge Sampling Model}\label{sec:edge:model}
We describe the sampling of edges through a random
process $\{H_i\}=\{H_i:i\in[|K|]\}$ where $H_i=1$ if $k_i$ is
selected $H_i=0$ otherwise.  Let $\cF_i$ denote the set of possible
outcomes $\{H_1,\ldots,H_i\}$;  We assume that an 
edge is selected according to a probability
that is a function of the sampling outcomes of previous edges. For
example, the selection probability of an edge can be a function of the
(random) number
of previously selected edges that are adjacent to it. Thus we write
\be
\Pr[k_i \mbox{ is selected } | \{H_1,\ldots,H_{i-1}\}] = \E[H_i|\cF_{i-1}]
  =p_i
\ee
where $p_i\in(0,1]$ is \textsl{random} probability that is determined
by the first $i-1$ sampling outcomes\footnote{Formally, $\{\cF_i\}$ is the natural filtration associated
with the process $\{H_i\}$, and $\{p_i\}$ is previsible
w.r.t. $\{\cF_i\}$; see \cite{W91}.}.

\subsection{Subgraph Estimation}

In this paper, we shall principally be concerned with estimating the
frequency of occurrence of certain subsets of $K$ within the sample. Our
principal tool is the \textbf{selection estimator} $\hat S_i =
H_i/p_i$ of the link $k_i$. It is uniquely defined by the properties:
(i) $\hat S_i \ge 0$;  (ii) $\hat S_i>0$ iff $H_i>0$; and (iii)
$\E[\hat S_i|\cF_{i-1}]=1$, which we prove in Theorem~\ref{thm:basic}
below. We recognize $\hat S_i$ as a
Horvitz-Thompson estimator \cite{HT52} of unity, which 
indicating the presence of $k_i$ in $K$.

The idea generalizes to indicators of general subsets of edges with
$K$. We call a subset $J\subset K$ an ordered subset when written in
increasing arrival order $J=(j_{i_1},j_{i_2},\ldots,j_{i_m})$ with
$i_1<i_2<\dots<i_m$. For an ordered subset $J$ of $K$ we write
\be
H(J)=\prod_{j_i\in J}H_i\quad\mbox{and}\quad P(J)=\prod_{j_i\in J}p_i
\ee
with the convention that $H(\emptyset)=P(\emptyset)=1$. We say that $J$ is
selected if $H(J)=1$. The selection estimator for an ordered subset
$J$ of $K$ is
\be
\hat S(J)=\prod_{j_i\in J}\hat S_{j_i}=H(J)/P(J)
\ee

Our main structural result concerns the properties if the $\hat S(J)$.

\begin{theorem}\label{thm:basic}
\begin{itemize}
\item[(i)] $\E[\hat S_i|\cF_{i-1}]=1$ and hence $\E[\hat S_i]=1$.
\item[(ii)] For any ordered subset $J=(j_{i_1},\ldots,j_{i_m})$ of $K$, 
\be
\E[\hat S(j_{i_1},\ldots,j_{i_m})|\cF_{i_{m-1}}]=\hat
S(j_{i_1},\ldots,j_{i_{m-1}})
\ee
and hence
\be
E[\hat S(J)]=1
\ee
\item[(iii)] Let $J,J'$ be two ordered subsets of $K$. 
If $J\cap J'=\emptyset$ then 
\be
\E[\hat S(J)\hat S(J')]=1\, \mbox{and hence}\,
\cov(\hat S(J),\hat S(J'))=0\ee
\item[(iv)] Let $J_1,\ldots,J_\ell$ be disjoint ordered subsets of
  $K$. Let $q$ be a polynomial in $\ell$ variables that is linear in
  each of its arguments. Then $\E[q(\hat S(J_1),\ldots,\hat S(J_\ell))]=q(1,\ldots,1)$.
\item[(v)] Let $J,J'$ be two ordered subsets of $K$ with $J \Delta
  J'$ their symmetric difference. Then $C(J,J')$ defined below is 
non-negative and an unbiased estimator of $\cov(\hat S(J),\hat S(J))$,
which is hence non-negative. $C(J,J')$ is defined to be $0$ when 
$J\cap J'=\emptyset$, and otherwise:
\be
\hat C(J,J')=\hat S(J\cup J')\left(\hat S(J\cap J')-1\right)
\ee
\item[(vi)] $\hat S(J)\left(\hat S(J)-1\right)$ is an unbiased estimator of $\var(\hat S(J))$.
\end{itemize}
\end{theorem}

\begin{proof}
(i)  $\E[\hat S_i|\cF_{i-1}]= \E[H_i/p_i|\cF_{i-1}]=1$, since $p_i>0$. 

(ii) is a corollary of (i) since 
\bea
&&\kern -20pt \E[\hat S(j_{i_1},\ldots,j_{i_m})|\cF_{i_{m-1}}]\\
&=&\E\left[\E[\hat
S_{i_m}|\cF_{i_m-1}]\hat S(j_{i_1},\ldots,j_{i_{m-1}})|\cF_{i_{m-1}}\right]\nonumber \\
&=&\hat S(j_{i_1},\ldots,j_{i_{m-1}})
\eea

(iii) When $J\cap J'=\emptyset$, then by (ii)
\be \E[\hat S(J)\hat S(J')]=\E[\hat S(J\cap J')]=1\ee

(iv) Is a direct corollary if (iii)

(v) Unbiasedness: The case $J\cap J'=\emptyset$ follows from (iii). Otherwise,
\bea
\E[\hat C(J,J)]&=&\E[\hat S(J)\hat S(J')]-\E[\hat S(J\cup J')]\\
&=&\E[\hat S(J)\hat S(J')]-1=\cov(\hat S(J),\hat S(J')]
\eea
 since $\E[\hat
S(J)]=\E[\hat S(J')]=1$. Nonnegativity: since each $\hat S(J)$ is
non-negative, $\hat C(J,J')$ is a product of $S(J\Delta J')$, which is
non-negative, with $H(J\cap J')(1/P^2(J\cap J')-1/P(J\cap J'))\ge 0$.

(vi) is a special case of (v) with $J=J'$.
\end{proof}

\section{Subgraph Sum Estimation}\label{sec-3}

We now describe in more detail the process of estimation, and
computing variance estimates. The most general quantity that we
wish to estimate is a weighted sum over collections of subgraphs; for
brevity, we will refer to these as \textbf{subgraph sums}.
 This class includes quantities such as counts of total nodes or links in
$G$, or counts of more complex objects such as connected paths of
length two, or triangles that have been a focus of study in the recent
literature. However, the class is more general quantities in which
a selector is applied to all subgraphs of a given type
(e.g. triangles) and only subgraphs fulfilling a selection criterion
(e.g. based on labels on the nodes of the triangle) are to be included
in the count. 

\subsection{General Estimation and Variance}

To allow for the greatest possible generality, we let $\cK=2^K$ denote
the set of subsets of $K$, and let $f$ be a real function on
$\cK$. For any subset $Q\subset \cK$, the subset sum of $f$ over $Q$
is 
\be
f(Q)=\sum_{J\in Q}f(J)
\ee
Here $Q$ represents the set of subgraphs fulfilling a selection
criterion as described above. Let $\hat Q$ denote the set of
objects in $Q$ that are sampled, i.e., those
$J=(k_{i_1},\ldots,k_{i_m})\in Q$ for which all links are selected.
The following is an obvious consequence
of the linearity of expectation and Theorem~\ref{thm:basic} 
\begin{theorem}
\begin{itemize}
\item[(i)] An unbiased estimator of $f(Q)$ is
\be\hat f(Q) = \sum_{J\in Q}f(J)\hat S(J)=\sum_{J\in\hat Q} f(J)/P(J)\ee 
\item[(ii)] An unbiased estimator of $\var(\hat f(Q))$ is
\be\label{eq:sum:cov}\nonumber
\sum_{J,J'\in\hat Q: J\cap J'\ne\emptyset }f(J)f(J')(1/P(J\cup J'))(1/P(J\cap J') -1)
\ee
\end{itemize}
\end{theorem}
Note that the sum in (\ref{eq:sum:cov}) can formally be left
unrestricted since terms with non-intersecting $J,J'$ are zero due to
our convention that $P(\emptyset)=1$.

\subsection{Edges}
As before $K$ denotes the edges in $G$;let $\hat K$ denote the set of
sampled edges. Then
\be\hat N_K =\sum_{k_i\in \hat K}\frac{1}{p_i}
\ee 
is an unbiased estimate of $N_K=|K|$. An unbiased estimate of the
variance of $\hat N_K$ is 
\be \sum_{k_i\in \hat K} \frac{1}{p_i} \left( \frac{1}{p_i} -1\right)
\ee

\subsection{Triangles}
Let $T$ denote the set of triangle $\tau=(k_1,k_2,k_3)$ in $G$, and
$\hat T$ the set of sampled triangles. Then
\be\hat N_T= \sum_{\tau\in \hat T}1/P(\tau)
\ee is an unbiased estimate of $N_T=|T|$, the number of triangles in $G$. Since
two intersecting triangles have either one link in common or are identical, an unbiased
estimate of $\var(\hat N_T)$ is
\be \sum_{\tau\in \hat T}\frac{1}{P(\tau)}\left(\frac{1}{P(\tau)}-1\right) +
 \sum_{\tau\ne\tau'\in \hat T}\frac{1}{P(\tau\cup\tau)}\left(\frac{1}{P(e(\tau,\tau')}-1\right)\nonumber
\ee
where $e(\tau,\tau')$ is the common edge between $\tau$ and $\tau'$
 
\subsection{Connected Paths of Length 2}
Let $\Lambda$ denote the set of connected paths of length two
$L=(k_1,k_2)$ in $G$, and $\hat\Lambda$ the subset of these that are sampled.
Then
\be\hat N_\Lambda= \sum_{L\in \hat \Lambda}1/P(L)
\ee is an unbiased estimate of $N_\Lambda=\vert \Lambda\vert$, the number of
such paths in $G$. Since
two non-identical members of $\Lambda$ have one edge in common,
an unbiased
estimate of $\var(\hat N_\Lambda)$ is
\be \sum_{L\in \hat \Lambda}\frac{1}{P(L)}\left(\frac{1}{P(L)}-1\right) +
 \sum_{L\ne L'\in \hat \Lambda}\frac{1}{P(L\cup L')}\left(\frac{1}{P(e(L,L')}-1\right)
\nonumber \ee
where $e(L,L')=L\cap L'$ is the common edge between $L$ and $L'$

\subsection{Clustering Coefficient}

The global clustering coefficient of a graph is defined as
$\alpha=3N_T/N_\Lambda$. While $3\hat N_T /\hat N_\Lambda$ is \textsl{an}
estimator of $\alpha$, it is not unbiased. However, the well known delta-method \cite{schervish} suggests using a formal Taylor expansion. But we note that a rigorous application of this method depends on
establishing asymptotic properties of $\hat N_T$ and $\hat N_\Lambda$
for large graphs, the study of which we defer to a subsequent
paper. With this caveat we proceed as follows. For a random vector
$X=(X_1,\ldots,X_n)$ a second order Taylor expansion results in the approximation
\be
\var(f(X_1,\ldots,X_n)) \approx v\cdot Mv
\ee
where $v=(\nabla f)(\E[X])$ and $M$ is the covariance matrix of the
$X_i$. Considering $f(\hat N_T,\hat N_\Lambda)=\hat N_T/\hat N_\Lambda$ we obtain the approximation:
\bea\label{eq:var:cluster}
\var (\hat N_T/\hat N_\Lambda)&\approx& 
\frac{\var(\hat N_T)}{N_\Lambda^2}+\frac{N_T^2 \var(\hat N_\Lambda )}{N_\Lambda^4} \\&&-2\frac{N_T\cov(\hat N_T,\hat N_\Lambda)}{N_\Lambda^3}
\eea
For computation we replace all quantities by their corresponding
unbiased estimators derived previously. Following
Theorem~\ref{thm:basic}, the covariance term is
estimated as
\be
\sum_{\tau\in\hat T,L\in\hat\Lambda \atop \tau\cap L\ne\emptyset}
\frac{1}{P(\tau\cup L)}\left(\frac{1}{P(\tau\cap L)}-1\right)
\ee

\subsection{Nodes}

Node selection is not directly expressed as a subgraph sum, but rather
through a polynomial of the type treated in Theorem~\ref{thm:basic}(iv).
Let $K(x)$ denote the edges containing the node $x\in V$. 
Now observe $x$ remains unsampled if and only if no edge in $K(x)$ is sampled.
This motivates the following estimator of node selection:
\be
\hat n_x=1-\prod_{k_i\in K(x)}(1-\hat S_i)
\ee
The following is a direct consequence of Theorem~\ref{thm:basic}(iv)
\begin{lemma}
$\hat n_x=0$ if and only if no edge from $K(x)$ is sampled, and
$\E[n_x]=1$.
\end{lemma}


\vspace{16mm}
\def\gsh{gSH}
\section{Graph Sample and Hold}\label{sec-4}

\begin{table*}[t!]
\begin{center}
\begin{tabular}{|ccc|ccc|c|ccc|cccc|}
\hline
\multicolumn{3}{|c|}{Order}&\multicolumn{3}{|c|}{Selection}&Prob.&\multicolumn{3}{|c|}{Weights}&\multicolumn{4}{|c|}{Est. Node Degree}\\
${}_{(a,b)}$ &${}_{(b,c)}$ &${}_{(c,d)}$ & ${}_{(a,b)}$ &${}_{(b,c)}$ &${}_{(c,d)}$&&${}_{(a,b)}$ &${}_{(b,c)}$ &${}_{(c,d)}$&${}_a$&${}_b$&${}_c$&${}_d$\\
\hline
1&2&3& \checkmark&\checkmark&\checkmark&$p$&$1/p$&1&1&$1/p$&$1/p+1$&2&1\\
 &&&  $\cdot$ &\checkmark&\checkmark&$(1-p)p$ & 0 &$1/p$ & 1&0&$1/p$&$1/p+1$&1\\
&&&   $\cdot$&$\cdot$&\checkmark&$(1-p)^2p$ & 0 &0 & $1/p$&0&0&$1/p$&$1/p$\\
&&&   $\cdot$&$\cdot$&$\cdot$&$(1-p)^3$ & 0 &0 & 0&0&0&0&0\\
\hline
2&1&3&\checkmark&\checkmark&\checkmark&$p$&1&$1/p$&1&1&$1/p+1$&$1/p+1$&1\\
&&&\checkmark&$\cdot$&\checkmark&$(1-p)p^2$&$1/p$&0&1/p&$1/p$&$1/p$&$1/p$&$1/p$\\
&&&$\cdot$&$\cdot$&\checkmark&$(1-p)^2p$&0&0&$1/p$&0&0&$1/p$&$1/p$\\
&&&\checkmark&$\cdot$&$\cdot$&$(1-p)^2p$&$1/p$&0&0&$1/p$&$1/p$&0&0\\
&&&   $\cdot$&$\cdot$&$\cdot$&$(1-p)^3$ & 0 &0 & 0&0&0&0&0\\
\hline
1&3&2&\checkmark&\checkmark&\checkmark&$p^2$&$1/p$&1&$1/p$&$1/p$&$1/p+1$&$1/p+1$&$1/p$\\
&&&\checkmark&\checkmark&$\cdot$&$p(1-p)$&$1/p$&1&0&$1/p$&$1/p+1$&1&0\\
&&&$\cdot$&\checkmark&\checkmark&$(1-p)p$&0&1&$1/p$&0&1&$1/p+1$&1\\
&&&$\cdot$&\checkmark&$\cdot$&$(1-p)^2p$&0&$1/p$&0&0&$1/p$&$1/p$&0\\
&&&  $\cdot$ &$\cdot$&$\cdot$&$(1-p)^3$ & 0 &0 & 0&0&0&0&0\\
\hline
\end{tabular}
\caption{Estimation on a path of length 3 using \gsh$(p,1)$}\label{tab:linear}
\end{center}
\end{table*}

\subsection{Algorithms}
We now turn to specific sampling algorithms that conform to the edge sampling model of Section~\ref{sec:edge:model}. \textbf{Graph Sample
  and Hold} \gsh$(p,q)$
 is a single pass algorithm over a stream of edges. The edge
 $k$ is somewhat analogous to the key of (standard) sample and hold. 
However, the notion of key matching is different. An arriving edge is deemed to match an edge currently stored if either of its nodes match a node currently being stored (in appropriate senses for the directed and undirected case). A matching edge is sampled with probability $q$. 
If there is not a match, the edge is stored with some probability $p$.
An edge not sampled is discarded permanently. For estimation purposes
we also need to keep track of the probability with which as selected
edge is sampled. We formally specify \gsh$(p,q)$ as Algorithm~\ref{alg:select}.

\begin{algorithm}
     \caption{Graph Sample and Hold: \gsh$(p,q)$}
     \label{alg:select}
     \SetKw{Return}{return}
     \SetVline \dontprintsemicolon 
     \BlankLine
     $\hat K \assign\emptyset$\;
     \While{new weighted edge $k$}{
       \If{$k\sim k'$ some $(k',p')\in \hat K$}{$r=q$\;}
       \Else{$r=p$\;}
       Append $(k,r)$ to $\hat K$ with probability $r$\;
}
\end{algorithm}

In some sense, \gsh\ samples connected components in
the same way the standard sample and hold samples flows, although
there are some differences. The main difference is a single connected
component in the original graph may be sampled as multiple
components. This can happen, for example, if omission of an edge from
the sample can disconnect a component. Clearly, the order in which
nodes are streamed determines whether or not such sampling
disconnection can occur.

Clearly, \gsh\ would admit generalizations that allow a more complex
dependence of sampling probability for new edge on the current sampled edge
set. Just as with \gsh\ itself, the details of the sampling scheme
should allow to certain subgraphs to be favor for selection. In this
paper we do not delve into this matter in great detail, rather we look
at a simple illustrative modification of \gsh\ that favor the selection
of triangles. \gsh$_T$ is identical to $\gsh$, except that any arriving edge that 
would complete a triangle is selected with probability 1; see 
 Algorithm~\ref{alg:select:ft}. Obviously \gsh$(p,1)$ and 
\gsh$_T(p,1)$ are identical.

\begin{algorithm}
     \caption{Graph Sample and Hold for Triangles: \gsh$_T(p,q)$}
     \label{alg:select:ft}
     \SetKw{Return}{return}
     \SetVline \dontprintsemicolon 
     \BlankLine
     
     $\hat K \assign\emptyset$\; 
     \While{new weighted edge $k$}{
      \If{$k$ would complete a triangle in $\hat K$}{$r=1$\;}
       \Else{\If{$k\sim k'$ some $(k',p')\in \hat K$}{$r=q$\;}
       \Else{$r=p$\;}}
       Append $(k,r)$ to $\hat K$ with probability $r$\;
}
\end{algorithm}

\subsection{Illustration with \gsh(p,1)}

We use a simple example of a path of length 3 to illustrate that in Graph Sample and
Hold \gsh$(p,1)$, the distribution of the random graph depends on the order in
which the edges are presented. The graph $G=(V,K)$ comprises 4 nodes
$V={a,b,c,d}$ connected by $3$ undirected edges
$K=\{(a,b),(b,c),(c,d)$ which are the keys for our setting. There are
6 possible arrival orders for the keys, of which we need only analyze
$3$, the other orders being obtained by time reversal. These are
displayed in the ``Order'' columns in Table~\ref{tab:linear}. For each
order, the possible selection outcomes for the three edges by the
check marks $\checkmark$, followed by the probability of each
selection. The adjusted weights for each outcome is displayed in
``Weights'' followed by corresponding estimate of the node degree,
i.e. the sum of weights of edges incident at each node. One can check by inspection that the probability-weighted sums of the weight estimators are $1$, while the corresponding sums of the degree estimators yield the the true node degree.

\section{Experiments and Evaluation}
\label{sec-experiments}
\begin{table}[t!]
\parbox[c]{.5\textwidth}{
\begin{center}
\caption{Statistics of datasets. $n$ is the number of nodes, $N_K$ is the number of edges, $N_T$ is the number of triangles, $N_\Lambda$ is the number of connected paths of length 2, $\alpha$ is the global clustering coefficient, and $D$ is the density.}
\vspace{-1.mm}
\label{tab:data-desc}
\scalebox{0.82}{
\begin{tabular}{ccccccccc}
\toprule
graph & $n$ & $N_K$ & $N_T$ & $N_\Lambda$ & $\alpha$ & $D$ \\
\midrule
socfb-CMU & 7K & 249.9K & 2.3M & 37.4M & 0.18526 & 0.0114 \\ 
socfb-UCLA & 20K & 747.6K & 5.1M & 107.1M & 0.14314 & 0.0036 \\
socfb-Wisconsin & 24K & 835.9K & 4.8M & 121.4M & 0.12013 & 0.0029 \\
\midrule
web-Stanford & 282K & 1.9M & 11.3M & 3.9T & 0.00862 & $5.01\times 10^{-5}$ \\
web-Google & 876K & 4.3M & 13.3M & 727.4M & 0.05523 & $1.15\times 10^{-5}$ \\
web-BerkStan & 685K & 6.6M & 64.6M & 27.9T & 0.00694 & $2.83\times 10^{-5}$ \\
\bottomrule
\end{tabular}}
\end{center}}
\end{table}
\begin{table}[t!]
\parbox[c]{.49\textwidth}{
\begin{center}
\caption{Estimates of expected value, relative error, sample size, lower bounds, and upper bounds when sample size $\leq 40K$ edges, with sampling probability $p,q = 0.005$ for web-BerkStan, and $p= 0.005, q = 0.008$ otherwise. $SSize$ is the number of sampled edges, and $LB, UB$ are the $95$\% lower, and upper bound respectively.}
\vspace{-2.mm}
\label{tab:res_40K}
\scalebox{0.79}{
\begin{tabular}{ccccccc}
\toprule
& \multicolumn{5}{c}{\textbf{Edges} $N_K$} & \\
\midrule
& $N_K$ & $\hat N_K$ & $\frac{|\hat N_K-N_K|}{N_K}$ & $SSize$ & $LB$ & $UB$ \\
\bottomrule
socfb-CMU & 249.9K &249.6K &0.0013 &1.7K &236.8K &262.4K\\
socfb-UCLA & 747.6K &751.3K &0.0050 &5K &729.3K &773.34K\\
socfb-Wisconsin & 835.9K &835.7K &0.0003 &5.5K &812.2K &859.1K\\
web-Stanford & 1.9M &1.9M &0.0004 &14.8K &1.9M &2M\\
web-Google & 4.3M &4.3M &0.0007 &25.2K &4.2M &4.3M\\
web-BerkStan & 6.6M &6.6M &0.0006 &39.8K &6.5M &6.7M\\
\bottomrule
\toprule
& \multicolumn{5}{c}{\textbf{Triangles} $N_T$} & \\
\midrule
& $N_T$ & $\hat N_T$ & $\frac{|\hat N_T-N_T|}{N_T}$ & $SSize$ & $LB$ & $UB$ \\
\bottomrule
socfb-CMU & 2.3M &2.3M &0.0003 &1.7K &1.6M &2.9M\\
socfb-UCLA & 5.1M &5.1M &0.0095 &5K &4.2M &6.03M\\
socfb-Wisconsin & 4.8M &4.8M &0.0058 &5.5K &4M &5.7M\\
web-Stanford & 11.3M &11.3M &0.0023 &14.8K &3.7M &18.8M\\
web-Google & 13.3M &13.4M &0.0029 &25.2K &11.7M &15M\\
web-BerkStan & 64.6M &65M &0.0063 &39.8K &45.5M &84.6M\\
\bottomrule
\toprule
& \multicolumn{5}{c}{\textbf{Path. Length two $N_\Lambda$}} & \\
\midrule
& $N_\Lambda$ & $\hat N_\Lambda$ & $\frac{|\hat N_\Lambda-N_\Lambda|}{N_\Lambda}$ & $SSize$ & $LB$ & $UB$ \\
\bottomrule
socfb-CMU & 37.4M &37.3M &0.0018 &1.7K &32.6M &42M\\
socfb-UCLA & 107.1M &107.8M &0.0060 &5K &100.1M &115.42M\\
socfb-Wisconsin & 121.4M &121.2M &0.0018 &5.5K &108.9M &133.4M\\
web-Stanford & 3.9T &3.9T &0.0004 &14.8K &3.6T &4.2T\\
web-Google & 727.4M &724.3M &0.0042 &25.2K &677.1M &771.5M\\
web-BerkStan & 27.9T &27.9T &0.0002 &39.8K &26.5T &29.3T\\
\bottomrule
\toprule
& \multicolumn{5}{c}{\textbf{Global Clustering $\alpha$}} & \\
\midrule
& $\alpha$ & $\hat{\alpha}$ & $\frac{|\hat{\alpha}-\alpha|}{\alpha}$ & $SSize$ & $LB$ & $UB$ \\
\bottomrule
socfb-CMU & 0.18526 &0.18574 &0.00260 &1.7K &0.14576 &0.22572\\
socfb-UCLA & 0.14314 &0.14363 &0.00340 &5K &0.12239 &0.16487\\
socfb-Wisconsin & 0.12013 &0.12101 &0.00730 &5.5K &0.10125 &0.14077\\
web-Stanford & 0.00862 &0.00862 &0.00020 &14.8K &0.00257 &0.01467\\
web-Google & 0.05523 &0.05565 &0.00760 &25.2K &0.04825 &0.06305\\
web-BerkStan & 0.00694 &0.00698 &0.00680 &39.8K &0.00496 &0.00900\\
\bottomrule
\end{tabular}}
\end{center}}
\end{table}
We test the performance of our proposed framework (gSH) as described in Algorithm~\ref{alg:select:ft} (with $r=1$ for edges that are closing triangles) on various social and information networks with $250K$--$7M$ edges. For all of the following networks, we consider an undirected graph, discard edge weights, self-loops, and we generate the stream by randomly permuting the edges. Table~\ref{tab:data-desc} summarizes the main characteristics of these graphs, such that $n$ is the number of nodes, $N_K$ is the number of edges, $N_T$ is the number of triangles, $N_\Lambda$ is the number of connected paths of length 2, $\alpha$ is the global clustering coefficient, and $D$ is the density. 
\begin{enumerate}
\item \textbf{Social Facebook Graphs}. Here, the nodes are people and edges represent friendships among Facebook users in three different US schools (CMU, UCLA, and Wisconsin)~\cite{traud2012social}.
\vspace{-2.mm}
\item \textbf{Web Graphs}. Here, the nodes are web-pages and edges are hyperlinks among these pages in different domains~\cite{leskovecRepository}.
\end{enumerate}
\vspace{-2.mm}
From Table~\ref{tab:data-desc}, we observe that social Facebook graphs are generally dense as compared to the web graphs.
\noindent
We ran the experiments on MacPro 2.66GHZ 6-Core Intel processor, with 48GB memory. In order to test the effect of parameter settings (i.e., $p$ and $q$), we perform $100$ independent experiments and we consider all possible combinations of $p$ and $q$ in the following range,
\begin{align*}
p,q =\{0.005, 0.008, 0.01, 0.03, 0.05, 0.08, 0.1\}
\end{align*}
Our experimental procedure is done as follows, independently for each $p=p_i,q=q_i$:
\vspace{-2.mm}
\begin{enumerate}
\item Given one parameter setting $p=p_i,q=q_i$, obtain a sample $\hat K$ using $gSH_T$($p_i$,$q_i$) (as in Algorithm~\ref{alg:select:ft})
\vspace{-2.mm}
\item Using S, compute the unbiased estimates of the following statistics:
  Edge counts $\hat N_K$; Triangle counts $\hat N_T$;  Connected paths of length two $\hat N_\Lambda$;  Global Clustering Coefficient $\hat \alpha$.
\item Compute the unbiased estimates of their variance
\end{enumerate}
\vspace{-2.mm}
\noindent

\subsection{Performance Analysis}
We proceed by first demonstrating how accurate the proposed framework's estimates for all the different graph statistics we discuss in this paper across various social and information networks. Given a sample $\hat{K} \subset K$, we consider the absolute relative error (i.e., $\frac{|E(est)-Actual|}{Actual}$) as a measure of how far is the estimate from the actual graph statistic of interest, where $E(est)$ is the mean estimated value across $100$ independent runs. Table~\ref{tab:res_40K} provides the estimates in comparison to the actual statistics when the sample size is $\leq 40K$ with $p,q=0.005$ for web-BerkStan and $p=0.005$, $q=0.008$ otherwise.   
We summarize below our main findings from Table~\ref{tab:res_40K}:
\vspace{-2.mm}
\begin{itemize}
\item For edge count estimates ($N_K$), we observe that the relative error is in the range of $0.03$\% -- $0.5$\% across all graphs.
\vspace{-2.mm}
\item For triangle count estimates ($N_T$), we observe that the relative error is in the range of $0.03$\% -- $0.95$\% across all graphs.
\vspace{-6.mm}
\item For estimates of the number of connected paths of length two ($N_\Lambda$), we observe that the relative error is in the range of $0.02$\% -- $0.6$\% across all graphs.
\vspace{-2.mm}
\item For global clustering coefficient estimates ($\alpha$), we observe that the relative error is in the range of $0.02$\% -- $0.76$\% across all graphs.
\vspace{-2.mm}
\item We observe that graphs that are more dense (such as socfb-UCLA) show higher error rates as compared to sparse graphs (such as web-Stanford).
\vspace{-2.mm}
\item From all above, we observe that the highest error is in the triangle count estimates and yet it is still $\leq 1$\%.
\end{itemize}
\vspace{-2.mm}

\begin{figure*}
\begin{center}
\subfigure[Edges]{\label{fig:socfb-UCLA-m}\includegraphics[width=0.33\linewidth]{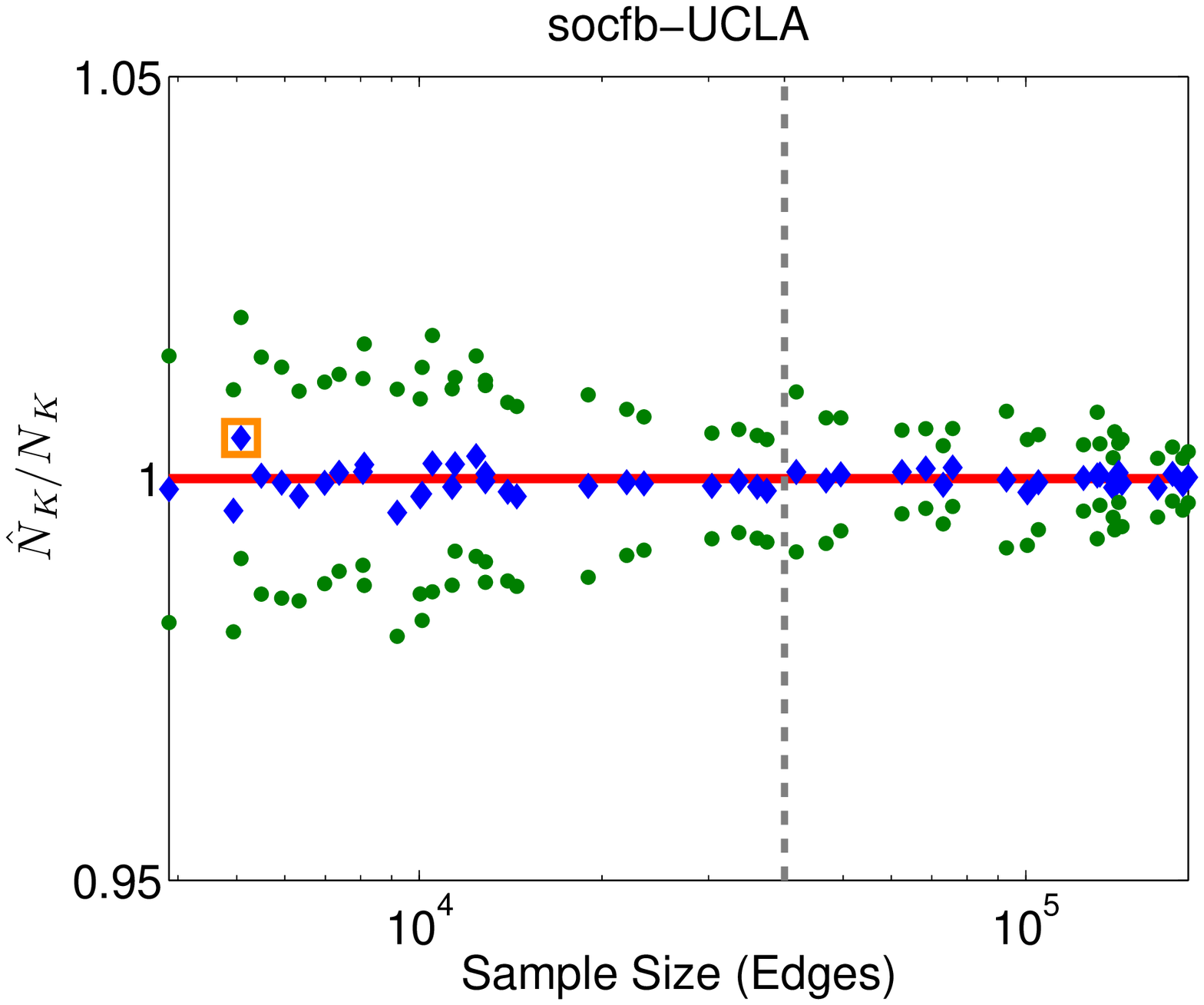}}
\subfigure[Triangles]{\label{fig:socfb-UCLA-T}\includegraphics[width=0.33\linewidth]{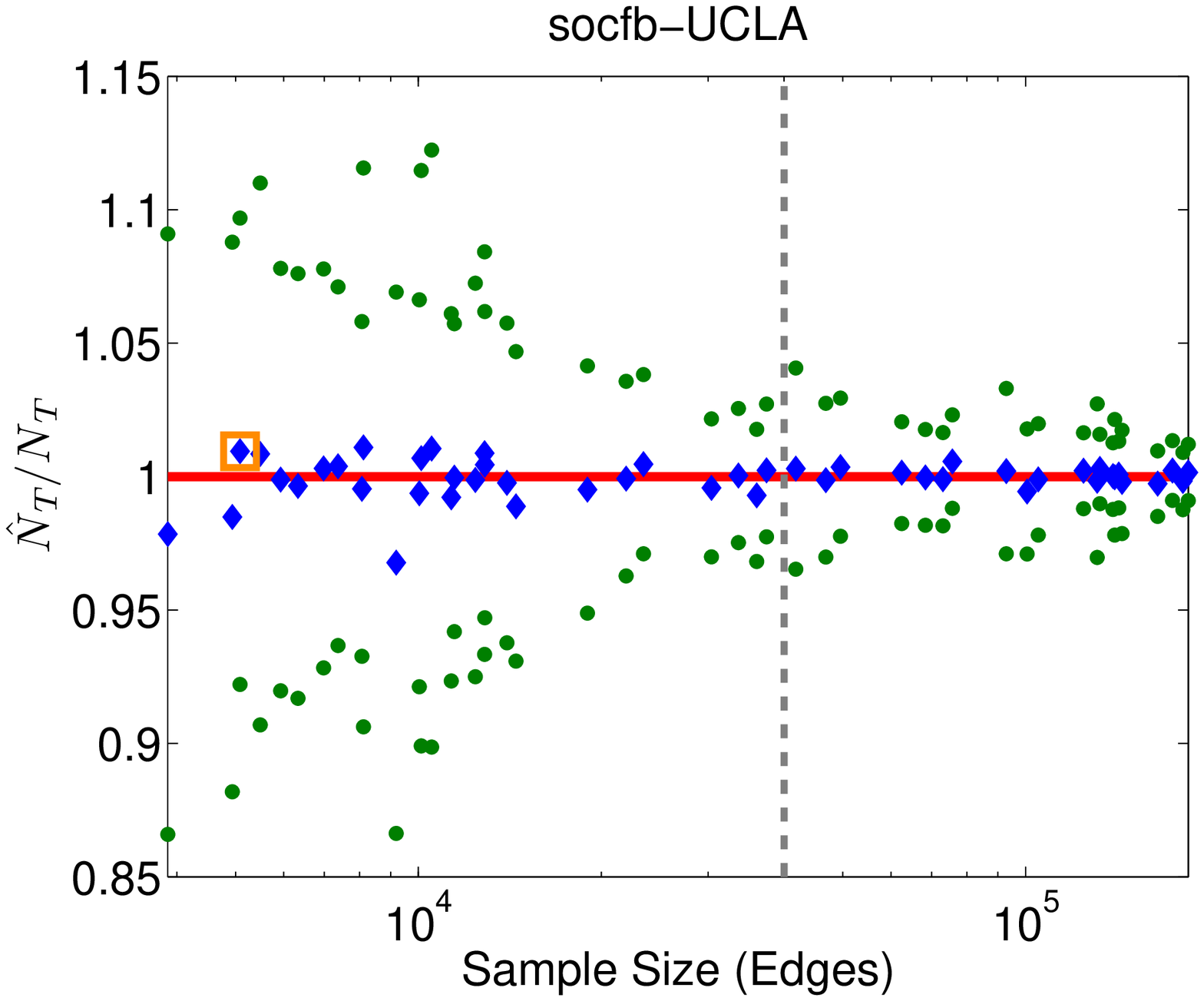}}
\subfigure[Path len.2]{\label{fig:socfb-UCLA-R}\includegraphics[width=0.33\linewidth]{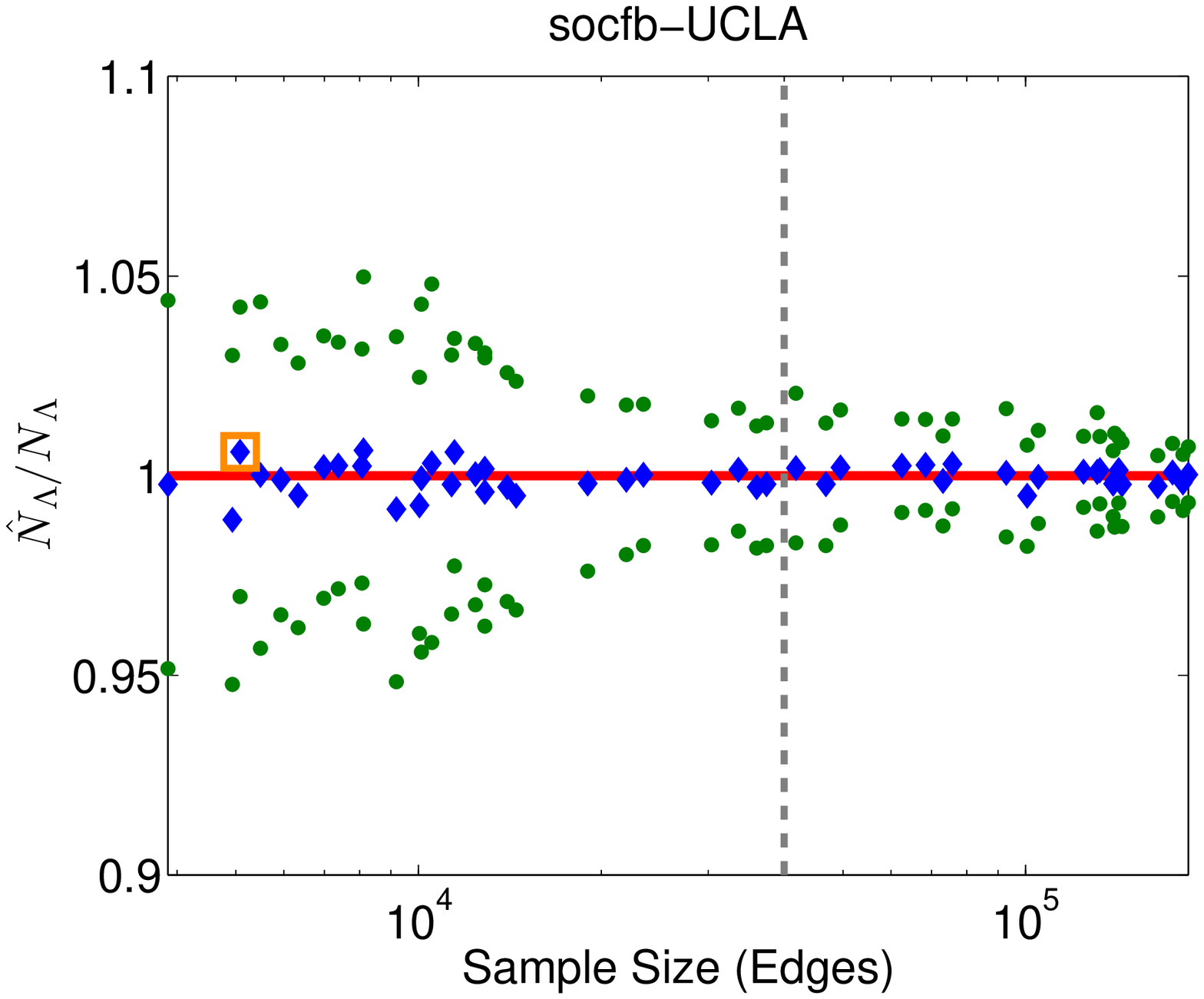}}
\vspace{-2.mm}
\subfigure[Edges]{\label{fig:socfb-Wisconsin87-m}\includegraphics[width=0.33\linewidth]{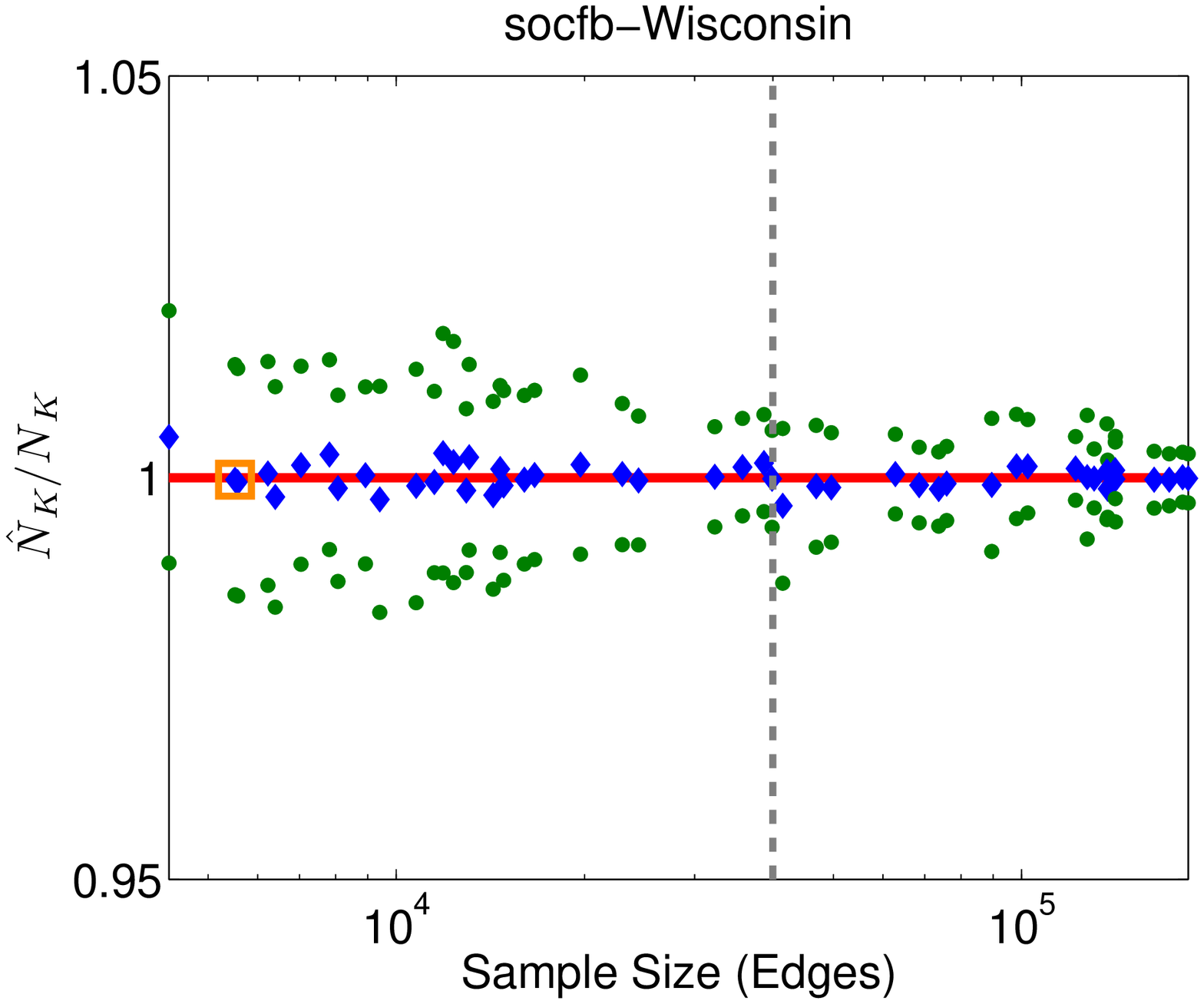}}
\subfigure[Triangles]{\label{fig:socfb-Wisconsin87-T}\includegraphics[width=0.33\linewidth]{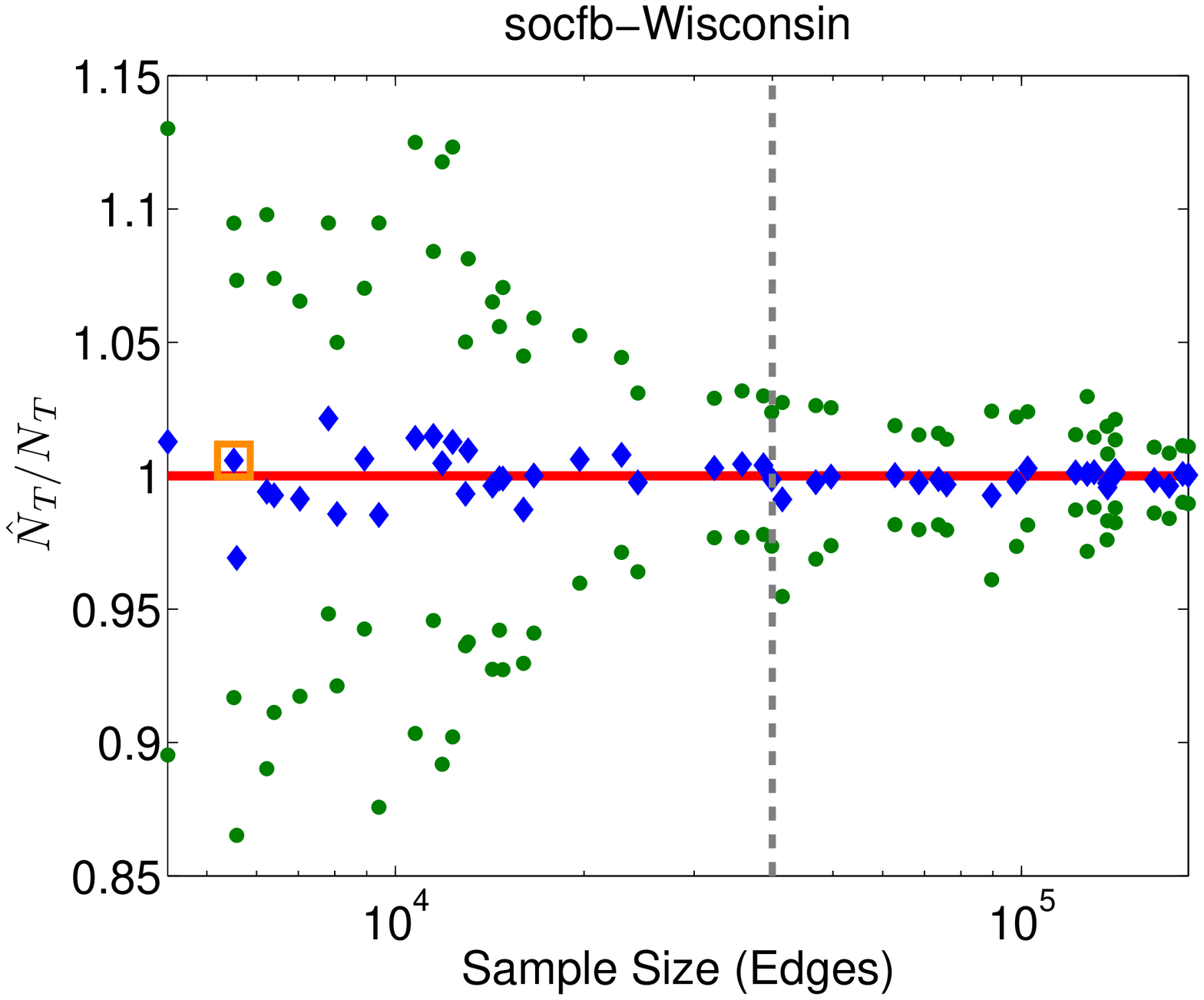}}
\subfigure[Path len.2]{\label{fig:socfb-Wisconsin87-R}\includegraphics[width=0.33\linewidth]{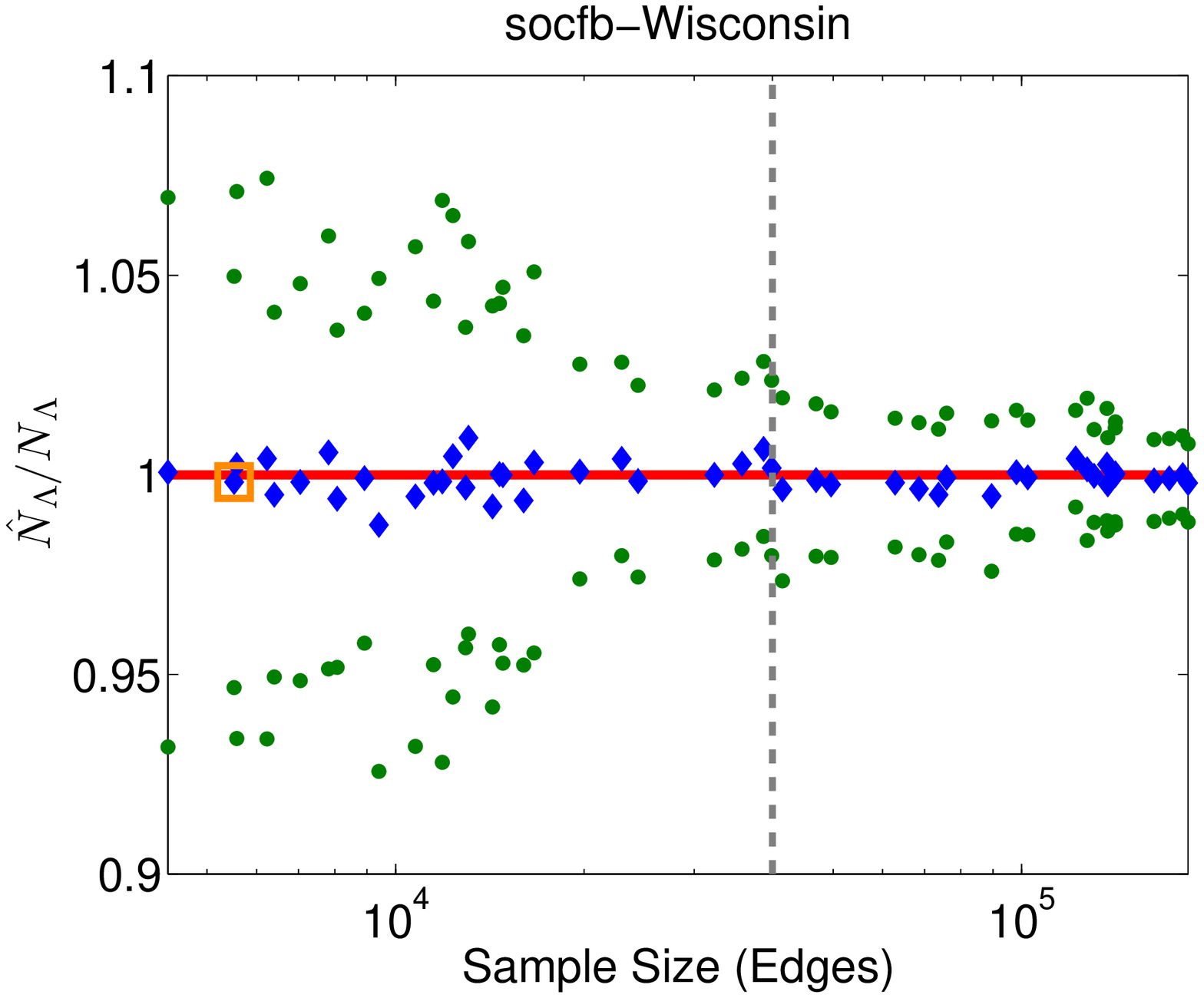}}
\vspace{-2.mm}
\subfigure[Clust.]{\label{fig:socfb-UCLA-k}\includegraphics[width=0.33\linewidth]{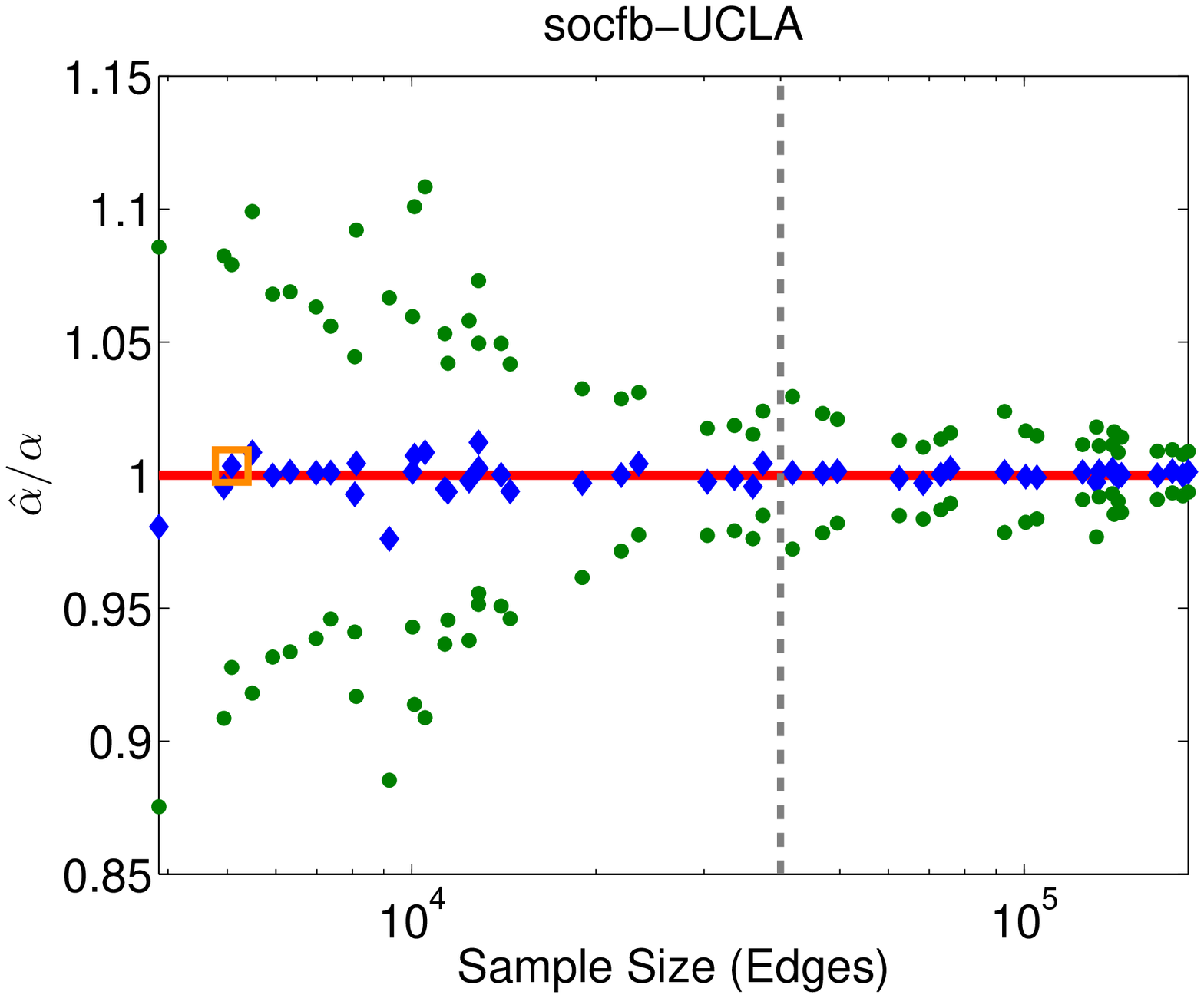}}
\subfigure[Clust.]{\label{fig:socfb-Wisconsin87-k}\includegraphics[width=0.33\linewidth]{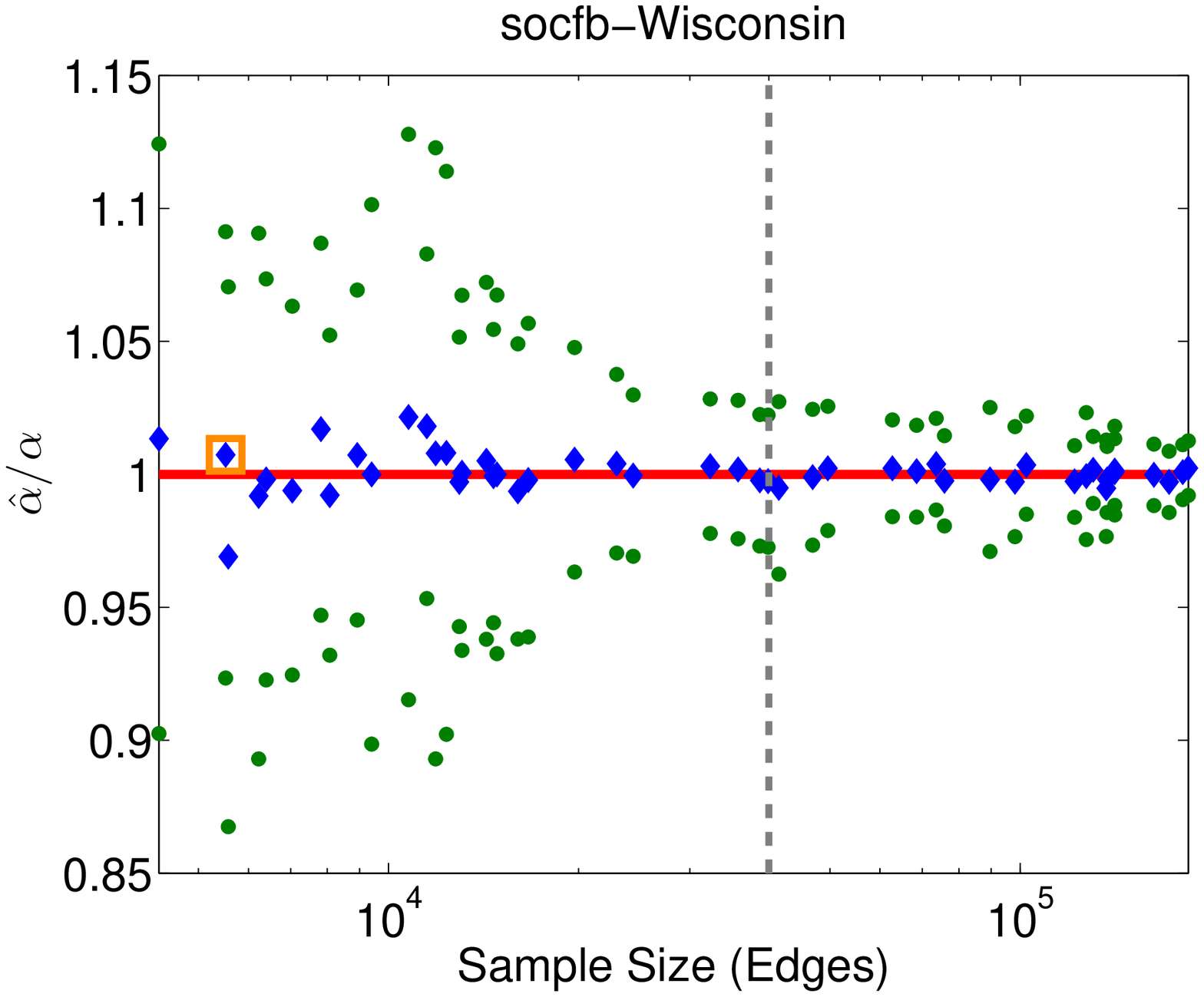}}
\vspace{-4.mm}
\caption{Convergence of the estimates ($N_K$, $N_T$, $N_\Lambda$, $\alpha$, upper and lower bounds) for socfb-UCLA and socfb-Wisconsin graphs, for all possible samples with $p,q$ in the range $0.005$--$0.1$.
Diamonds (Blue): $\frac{E(est)}{Actual}$. Circles (Green): $\frac{UB}{Actual}, \frac{LB}{Actual}$. Square (Gold): refers to the sample in Table~\ref{tab:res_40K}. Dashed line (Grey): refers to the sample with sample size = 40K edges}
\label{fig:convergence}
\end{center}
\end{figure*}

\begin{figure*}
\begin{center}
\subfigure{\label{fig:web-Google-srate}}\includegraphics[width=0.25\linewidth]{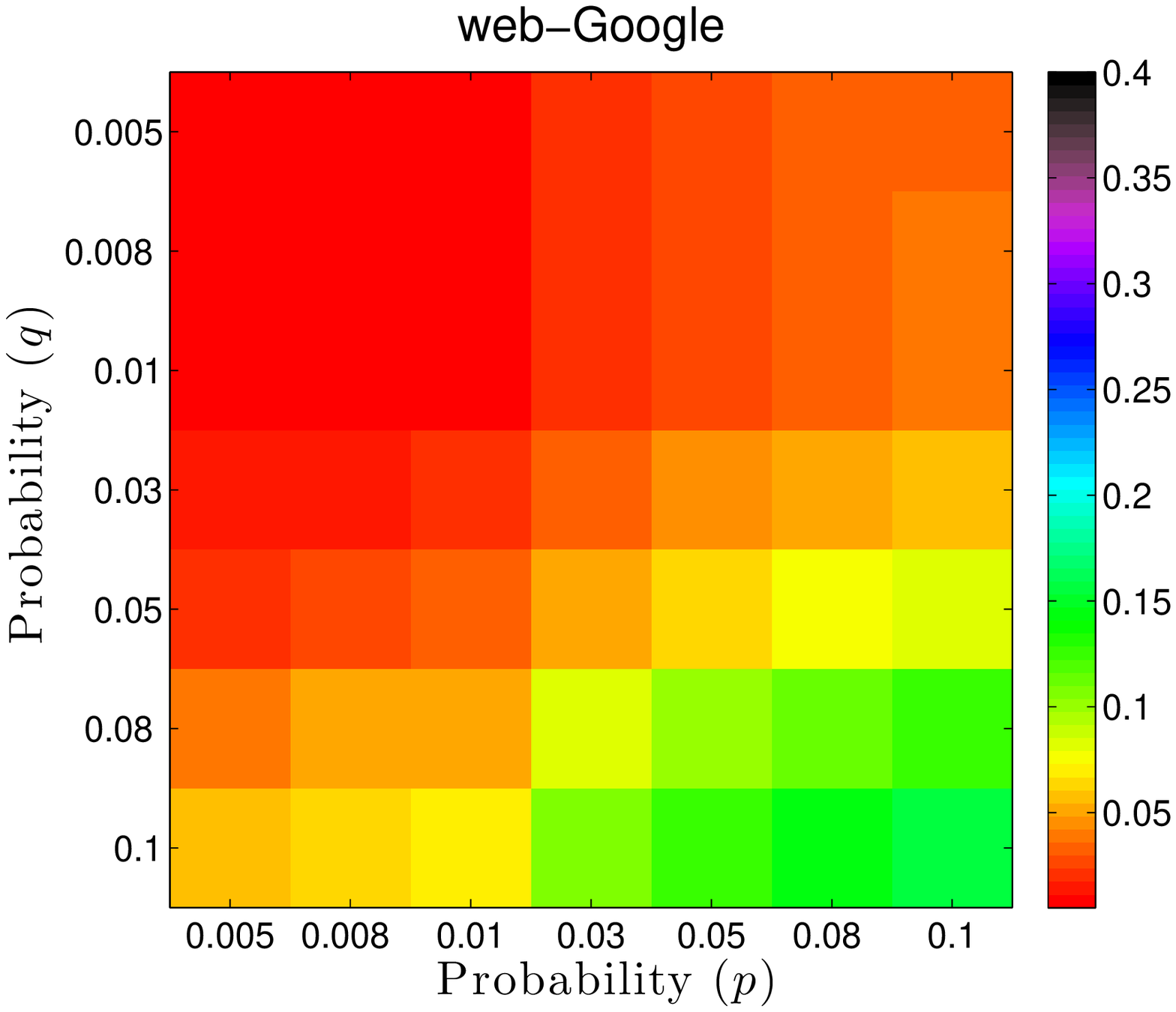}
\hspace{-2.mm}
\subfigure{\label{fig:web-Stanford-srate}}\includegraphics[width=0.25\linewidth]{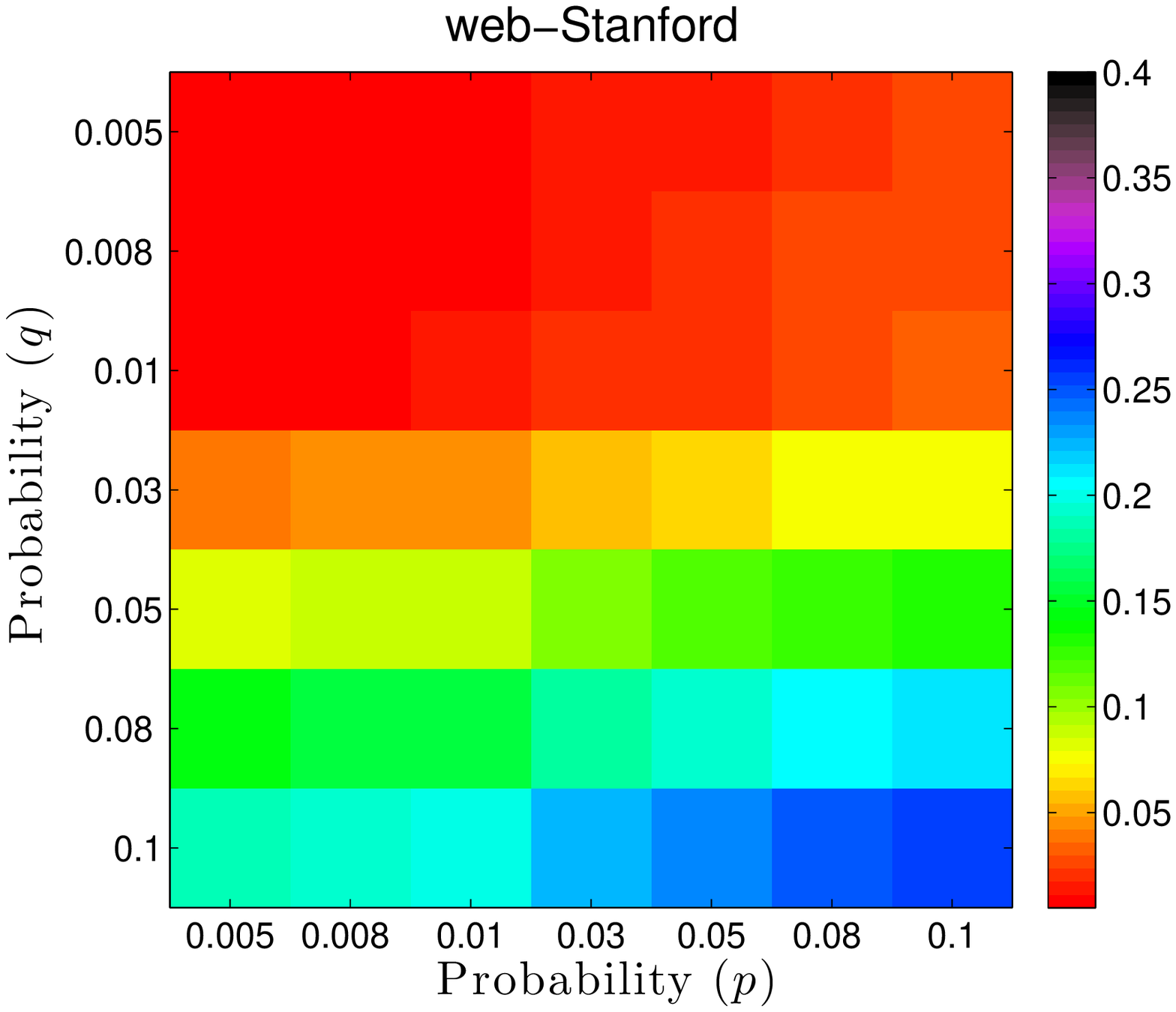}
\hspace{-2.mm}
\subfigure{\label{fig:socfb-Wisconsin87-srate}}\includegraphics[width=0.25\linewidth]{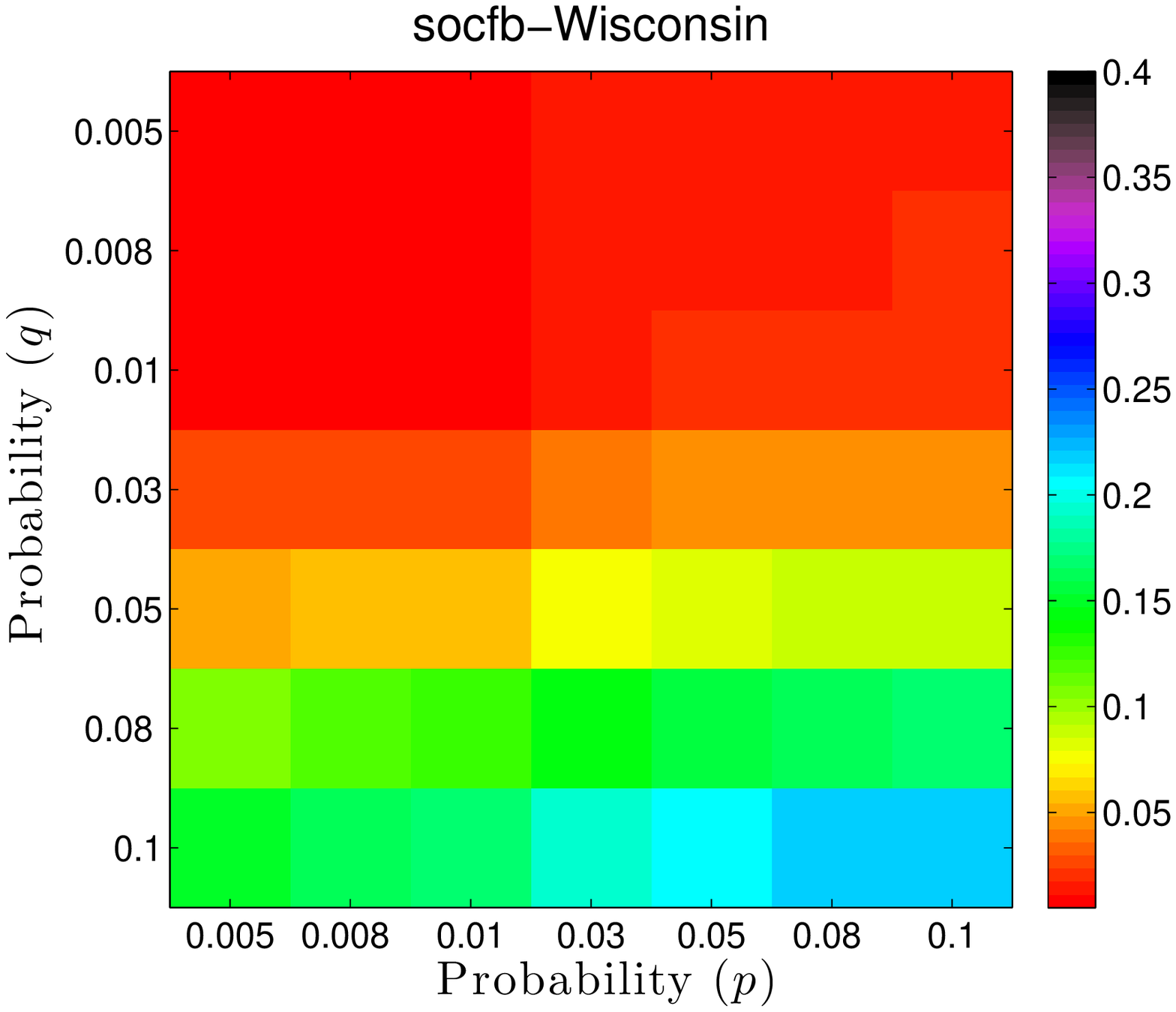}
\hspace{-2.mm}
\subfigure{\label{fig:socfb-CMU-srate}}\includegraphics[width=0.25\linewidth]{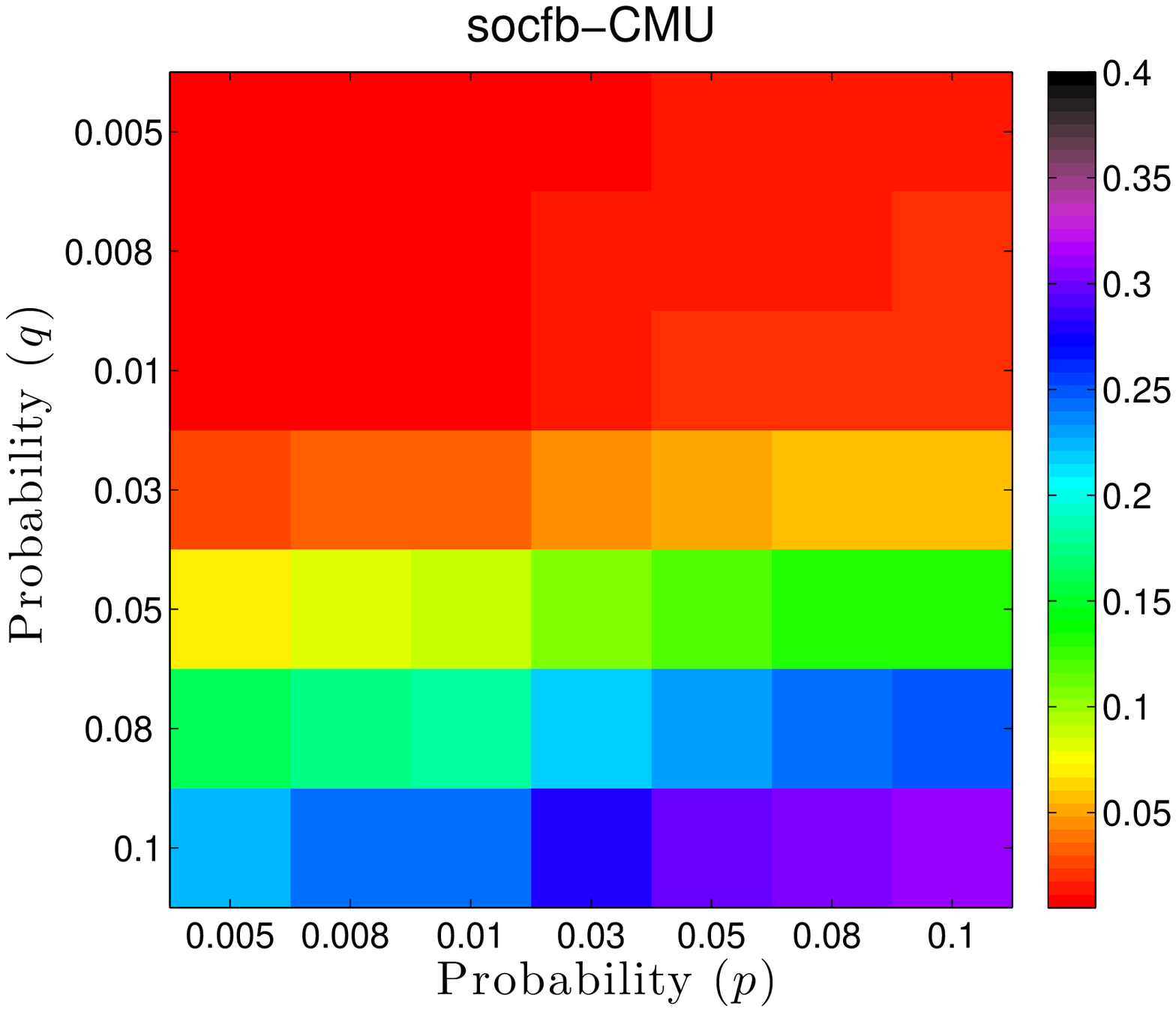}
\vspace{-4.mm}
\caption{Sampling Fraction ($\frac{SSize}{m}$, where $SSize$ is the number of sampled edges) as $p,q$ changes in the range $0.005$--$0.1$ for  web-Google, web-Stanford, socfb-Wisconsin, and socfb-CMU graphs (ordered from sparse $\rightarrow$ dense).}
\label{fig:sample-size}
\end{center}
\end{figure*}

\subsection{Confidence Bounds}
Having selected a sample that can be used to estimate the actual statistic, it is also desirable to construct a confidence interval within which we are \emph{sufficiently} sure that the actual graph statistic of interest lies. 
We construct a $95$\% confidence interval for the estimates for edge ($N_K$), triangle ($N_T$), connected paths of  length 2 ($N_\Lambda$) counts, and clustering coefficient ($\alpha$) as follows, 
\begin{align}
est \pm 1.96 \sqrt{Var(est)} 
\end{align}
where the estimates $est$ and $Var(est)$ are computed using the equations of the unbiased estimators of counts and their variance as discussed in Section~\ref{sec-framework}. For example, the $95$\% confidence interval for the  edge count is, 
\begin{align}
\hat N_K \pm 1.96 \sqrt{Var(\hat N_K)} 
\end{align}
where $UB = \hat N_K+1.96 \sqrt{Var(\hat N_K)}$, $LB = \hat N_K-1.96 \sqrt{Var(\hat N_K)}$ are the upper and lower bounds on the edge count respectively.
\noindent
Table~\ref{tab:res_40K} provides the $95$\% upper and lower bounds (i.e., $UB,LB$) for the sample when the sample size is $\leq 40K$ edges. We observe that the actual statistics across all different graphs lie in between the bounds of the confidence interval (i.e., $LB \leq Actual \leq UB$).

Additionally, we study the properties of the sampling distribution of our proposed framework (gSH) as we change the sample size. Figure~\ref{fig:convergence} shows the sampling distribution as we increase the sample size (for all possible settings of $p,q$ in the range $0.005$--$0.1$ as described previously). More specifically, we plot the fraction $\frac{E(est)}{Actual}$ (blue diamonds in the figure), where $E(est)$ is the mean estimated value across $100$ independent runs. Further, we plot the fractions $\frac{UB}{Actual}$, and $\frac{LB}{Actual}$ (green circles in the figure). These plots show the sampling distribution of all statistics for socfb-UCLA, and socfb-Wisconsin graphs.
\noindent
We now summarize our findings from Figure~\ref{fig:convergence}:
\vspace{-5.mm}
\begin{itemize}
\item We observe that the sampling distribution is centered and balanced over the red line ($y_{axis} =1$) which represents the actual value of the graph statistic. This observation shows the unbiased properties of the estimators for the four graph quantities of interest that we discussed in Section~\ref{sec-framework}. 
\vspace{-2.mm}
\item We observe that the upper and lower bounds contain the actual value (represented by the red line) for different combinations of $p,q$
\vspace{-2.mm}
\item We observe that as we increase the sample size, the bounds \emph{converge} to be more concentrated over the actual value of the graph statistic (i.e, variance is decreasing)
\vspace{-2.mm}
\item We observe that the confidence intervals for edge counts are small in the range of $0.98$--$1.02$
\vspace{-2.mm}
\item We observe that the confidence intervals for triangle counts and clustering coefficient are large compared to other graph statistics (in the range of $0.87$--$1.12$).
\vspace{-2.mm}
\item We observe that samples with size $= 40K$ edges provide a reasonable tradeoff between the sample size and unbiased estimates with low variance  
\vspace{-2.mm}
\item Thus we conclude that the sampling distribution of the proposed framework has many desirable properties of unbiasedness and low variance as we increase the sample size. 
\vspace{-2.mm}
\end{itemize}
Note that in Figure~\ref{fig:convergence}, we use a square (with gold color) to refer to the sample reported in Table~\ref{tab:res_40K}. We also found similar observations for the remaining graphs (omitted due to space constraints).

In addition to the analysis above, we compute the exact coverage probability $\gamma$ of the $95$\% confidence as follows,
\begin{align}
\gamma = P(LB \leq Actual \leq UB)
\end{align}

For each $p = p_i,q = q_i$, we compute the proportion of samples in which the actual statistic lies in the confidence interval across $100$ independent sampling experiments $gSH_T(p_i,q_i)$. We vary $p,q$ in the range of $0.005$--$0.01$, and for each possible combination of $p,q$ (e.g., $p=0.005, q=0.008$), we compute the exact coverage probability $\gamma$. Table~\ref{tab:conf-level} provides the mean coverage probability with $p,q = \{0.005, 0.008,0.01\}$ for all different graphs. Note $\gamma_{N_K}$, $\gamma_{N_T}$, $\gamma_{N_\Lambda}$, and $\gamma_{\alpha}$ indicate the exact coverage probability of edge, triangle, path length 2 counts, and clustering coefficient respectively. We see that the nominal $95$\% confidence interval holds to a good approximation, as $\gamma \approx 95$\% across all graphs.   

\begin{table}[t!]
\parbox[c]{.4\textwidth}{
\begin{center}
\caption{Coverage Probability $\gamma$ for $95$\% confidence interval \label{tab:conf-level}}
{
\begin{tabular}{ccccccccccc}
\toprule
graph & $\gamma_{N_K}$ &  & $\gamma_{N_T}$ & & $\gamma_{N_\Lambda}$ & & $\gamma_\alpha$ \\
\midrule
socfb-CMU & 0.94 &  & 0.95 & & 0.96 &  & 0.92\\
socfb-UCLA & 0.96 &  & 0.95 & & 0.95 & & 0.92 \\
socfb-Wisconsin & 0.95 & & 0.95 & & 0.96 & & 0.95 \\
web-Stanford & 0.97 & & 0.92 & & 0.95 & & 0.92 \\
web-Google & 0.95 & & 0.93 & & 0.95 & & 0.95 \\
web-BerkStan & 0.96 & & 0.94 & & 0.93 & & 0.93 \\
\bottomrule
\end{tabular}}
\end{center}}

\parbox[c]{.4\textwidth}{
\begin{center}
\caption{The relative error and sample size of Jha~\cite{jha2013space} in comparison to our framework for triangle count estimation \label{tab:comp_kdd13}}
\vspace{-1.mm}
{
\begin{tabular}{llllr}
\toprule
& \multicolumn{2}{c}{Jha \etal~\cite{jha2013space}} &  \multicolumn{2}{c}{gSH} \\
\cmidrule(r){2-3}
\cmidrule(r){4-5}
graph & $\frac{|\hat N_T-N_T|}{N_T}$ & $SSize$ & $\frac{|\hat N_T-N_T|}{N_T}$ & $SSize$ \\ 
\midrule
web-Stanford & $\approx 0.07$ & 40K & 0.0023  & 14.8K \\
web-Google &	 $\approx 0.04$  & 40K & 0.0029 & 25.2K \\
web-BerkStan & $\approx 0.12 $  & 40K & 0.0063 & 39.8K \\
\bottomrule
\end{tabular}}
\end{center}}
\end{table}
\subsection{Comparison to Previous Work}
We compare to the most recent research done on triangle counting by Jha \etal~\cite{jha2013space}. Jha \etal proposed a Streaming-Triangles algorithm to estimate the triangle counts. Their algorithm maintains two data structures. The first data structure is the edge reservoir and used to maintain a uniform random sample of
edges as they streamed in. The second data structure is the wedge (path length two) reservoir and used to select a uniform sample of wedges created by the edge reservoir. The algorithm proceeds in a reservoir sampling fashion as a new edge $e_t$ is streaming in. Then, edge $e_t$ gets the chance to be sampled and replace a previously sampled edge with probability $1/t$. Similarly, a randomly selected new wedge (formed by $e_t$) replaces a previously sampled wedge from the wedge reservoir. Table~\ref{tab:comp_kdd13} provides a comparison between our proposed framework (gSH) and the Streaming-Triangles algorithm proposed in ~\cite{jha2013space}. Note that we compare with the results reported in their paper.   

From Table~\ref{tab:comp_kdd13}, we observe that across the three web graphs, our proposed framework has a relative error \emph{orders of magnitude} less than the Streaming-Triangles algorithm proposed in ~\cite{jha2013space}, as well as with a small(er) overhead storage (in most of the graphs). We note that the work done by Jha \etal~\cite{jha2013space} compares to other state of the art algorithms and shows that they are not practical and produce a very large error; see Section~\ref{sec-related} for more details.

\subsection{Effect of $p,q$ on Sampling Rate}

While Figure~\ref{fig:convergence} shows that the sampling distribution of the proposed framework is unbiased regardless the choice of $p,q$, the question of what is the effect of the choice of $p,q$ on the sample size still needs to be explored. In this section, we study the effect of the choice of parameter settings on the fraction of edges sampled from the graph. 

Figure~\ref{fig:sample-size} shows the fraction of sampled edges as we vary $p,q$ in the range of $0.005$--$0.1$ for two web graphs and two social Facebook graphs. Note that the graphs are ordered by their density (check Table~\ref{tab:data-desc}) going from the most sparse to the most dense graph. We observe that when $q \leq 0.01$, regardless the choice of $p$, the fraction of sampled edges is in the range of $0.5$\% -- $2.5$\% of the total number of edges in the graph. We also observe that as $q$ goes from $0.01$ to $0.03$, the fraction of sampled edges would be in the range of $2.75$\% -- $5$\%. These observations hold for all the graphs we studied.

On the other hand, as $q$ goes from $0.03$ to $0.1$, the fraction of sampled edges depends on whether the graph is dense or sparse. For example, for web-Google graph, as $q$ goes from $0.03$ to $0.1$, the fraction of sampled edges goes from $5$\% to $15$\%. Also, for web-Stanford graph, as $q$ goes from $0.03$ to $0.1$, the fraction of sampled edges goes from $5$\% to $25$\%. Moreover, for the most dense graph we have in this paper (socfb-CMU), the fraction of sampled edges goes from $5$\% to $31$\%. Note that when we tried $q=1$, regardless the choice of $p$, at least more than $80$\% of the edges were sampled. 

Since $p$ is the probability of sampling a fresh edge (not adjacent to a previously sampled edge), one could think of $p$ as the probability of random jumps (similar to random walk methods) to explore unsampled regions in the graph. On the other hand, $q$ is the probability of sampling an edge adjacent to previous edges. Therefore, one could think of $q$ as the probability of exploring the neighborhood of previously sampled edges (similar to the forward probability in Forest Fire sampling~\cite{leskovec2006slg}).

\noindent
From all the discussion above, we conclude that using a small $p,q$ settings (i.e., $\leq 0.008$) is better to control the fraction of sampled edges, and also recommended since the sampling distribution of the proposed framework is unbiased regardless the choice of $p,q$ as we show in Figure~\ref{fig:convergence} (also see Section~\ref{sec-framework}). However, if a tight confidence interval is needed, then increasing $p,q$ helps reduce the variance estimates.

\begin{table}[t!]
\parbox[c]{.5\textwidth}{
\begin{center}
\caption{Elapsed time (seconds) for counting edges, triangles, and paths of len.2}
\label{tab:time}
\begin{tabular}{lcccc}
\toprule
& \multicolumn{2}{c}{Full Graph} &  \multicolumn{2}{c}{Sampled Graph} \\
\cmidrule(r){2-3}
\cmidrule(r){4-5}
graph & Time & Graph size & Time & SSize\\
\midrule
web-Stanford & 19.68  & 1.9M & 0.13 & 14.8K\\
web-Google & 5.05 & 4.3M & 0.55 & 25.2K \\
web-BerkStan & 113.9 & 6.6M & 1.05 & 39.8K\\
\bottomrule
\end{tabular}
\end{center}}
\end{table}
\subsection{Implementation Issues} 
\noindent
In practice, statistical variance estimators are costly to compute. In this paper, we provide an efficient parallel procedure to compute the variance estimate. We take triangles as an example. Consider for example any pair of triangles $\tau$ and $\tau'$, assuming $\tau$ and $\tau'$ are not identical, the covariance of $\tau$ and $\tau'$ is greater than zero, if and only if the two triangles are intersecting. Since two intersecting triangles have either one edge in common or are identical, we can find intersecting triangles by finding all triangles \emph{incident} to a particular edge $e$. In this case, the intersection probability of the two triangles is $P(\tau\cap\tau') = P(e)$. Note that if $\tau$ and $\tau'$ are identical, then the computation is straightforward.

\noindent
The procedure is very simple as follows,
\begin{itemize}
\item Given a sample set of edges $\hat K$, for each edge $e \in \hat{K}$
\begin{itemize}
\item find the set of all triangles ($T_e$) incident to $e$ 
\item for each pair $(\tau,\tau')$, where $\tau,\tau' \in T_e$. Compute the $Cov(\tau,\tau')$ such that $P(\tau\cap\tau') = P(e)$
\end{itemize}
\end{itemize}
\noindent
Since, the computation of each edge is independent of other edges, we parallelize the computation of the variance estimators. Moreover, since the computation of triangle counts and paths of length two can themselves be parallelized, we compare the total elapsed time in seconds used to compute these counts on both the full graph and a sampled graph of size $\leq 40K$ edges. Table~\ref{tab:time} provide the results of this comparison for the three web graphs. Note that in the case of the sampled graph, we also sum the computations of the variance estimators in addition to the triangle and paths of length two count estimators. Also, note that we use the sample reported in Table~\ref{tab:res_40K}. The results show a significant reduction in the time needed to compute triangles and paths of length two counts. For example, consider the web-BerkStan graph, where the total time is reduced from 113 seconds to 1.05 seconds. Note that all the computations of Table~\ref{tab:time} are performed on a Macbook Pro laptop 2.9GHZ Intel Core i7 with 8GB memory. Note that the storage state of gSH is only in terms of the number of sampled edges. In others words, the storage of the sampling probabilities is negligible since it is not part of the in-memory consulting state of the stream sampling framework gSH. Moreover, we use only three different probabilities, ($p,q$ and $1$), that can be stored with a custom $2-bit$ data structure, where $00$, $01$, and $10$ represents $p,q$ and $1$ respectively. 
\section{Related Work}
\label{sec-related}

In this section, we discuss the related work on the problem of large-scale graph analytics and their applications. Generally speaking, there are two bodies of work related to this paper: (i) graph analytics in graph stream setting, and (ii) graph analytics in the non-streaming setting (e.g. using \textsc{MapReduce} and \textsc{Hadoop}). In this paper, we propose a generic stream sampling framework for big-graph analytics, called Graph Sample and Hold (gSH), that works in a \emph{single pass} over the streams. Therefore, we focus on the related work for graph analytics in graph stream setting and we briefly review the other related work.

\paragraph{\textbf{\textit{Graph Analysis Using Streaming Algorithms}}}

Before exploring the literature of graph stream analytics, we briefly review the literature in data stream analysis and mining that may not contain graph data. For example, for sequence sampling (\eg, reservoir sampling)~\cite{Vitter:85,babcock2002sampling}, for computing frequency counts~\cite{manku,charikar2002finding} and load shedding~\cite{tatbul2003load}, and for mining concept drifting data streams~\cite{fan2004streamminer}. Additionally, The idea of \emph{sample and hold} (\textsc{SH}) was introduced in~\cite{Estan2002} for unbiased sampling of network measurements with integral weights. 
Subsequently, other work explored adaptive \textsc{SH}, and \textsc{SH} with signed updates~\cite{cohen2012don,cohen2007algorithms}. 
Nevertheless, none of this work has considered the framework of \emph{sample and hold} (\textsc{SH}) for social and information networks. In this paper, however, we propose the framework of \emph{graph sample and hold} (gSH) for big-graph analytics. 

There has been an increasing interest in mining, analysis, and querying of massive graph streams as a result of the proliferation of graph data (\eg, social networks, emails, IP traffic, Twitter hashtags). Following the earliest work on graph streams~\cite{raghavan1999computing}, several types of problems were explored in the field of analytics of massive graph streams. 
For example, to count triangles~\cite{jha2013space,pavan2013counting,yossef,buriol2006counting,becchetti2008efficient,jowhari2005new}, finding common neighborhoods~\cite{buchsbaum2003finding}, estimating pagerank values~\cite{atish}, and characterizing degree sequences in multi-graph streams~\cite{cormode2005space}. In the data mining and machine learning field, there is the work done on clustering graph streams~\cite{aggarwal2010clustering}, outlier detection~\cite{aggarwal2011outlier}, searching for subgraph patterns~\cite{chen2010continuous}, mining dense structural patterns~\cite{aggarwal2010dense}, and querying the frequency of particular edges and subgraphs in the graphs streams~\cite{zhao2011gsketch}. For an excellent survey on analytics of massive graph streams, we refer the reader to~\cite{mcgregor2009graph,zhang2010survey}.

Much of this work has used various sampling schemes to sample from the stream of graph edges. Surprisingly, the majority of this work has focused primarily on sampling schemes that can be used to estimate certain graph properties (e.g. triangle counts), while much less is known for the case when we need a generic approach to estimate various graph properties with the same sampling scheme with minimum assumptions.

For example, the work done in~\cite{buriol2006counting} proposed an algorithm with space bound guarantees for triangle counting and clustering estimation in the \emph{incidence stream model} where all edges incident to a node arrive in order together. However, in the incidence stream model, counting triangles is a relatively easy problem, and counting the number of paths of length two is simply straightforward. On the other hand, it has been shown that these bounds and accurate estimates will no longer hold in the case of \emph{adjacency stream model}, where the edges arrive arbitrarily with no particular order~\cite{jha2013space,pavan2013counting}.

Another example, the work done Jha \etal in~\cite{jha2013space} proposed a practical, single pass, $O(\sqrt{n})$-space streaming algorithm specifically for triangle counting and clustering estimation with additive error guarantee (as opposed to other algorithms with relative error guarantee). Although, the algorithm is practical and approximates the triangle counts accurately at a sample size of $40K$ edges, their method is specifically designed for triangle counting. Nevertheless, we compare to the results of triangle counts reported in~\cite{jha2013space}, and we show that our framework is not only generic but also produces errors with orders of magnitude less than the algorithm in~\cite{jha2013space}, and with a small(er) storage overhead in many times.

More recently, Pavan \etal proposed a space-efficient streaming algorithm for counting and sampling triangles in~\cite{pavan2013counting}. This algorithm is practical and works in a single pass streaming fashion with order $O(m\Delta/T)$-space, where $\Delta$ is the maximum degree of the graph. However, this algorithm needs to store estimators (i.e., wedges that may form potential triangles), and each of these estimators stores \emph{at least} one edge. In their paper, the authors show that they need at least $128$ estimators (i.e., more than $128K$ edges), to obtain accurate results (large storage overhead compared to this paper). The sampling algorithm of~\cite{pavan2013counting} bears some formal resemblance to our approach in using different sampling probabilities 
depending on whether or not an arriving edge is adjacent to a previous edge, but otherwise the details are substantially different. 

Other semi-streaming algorithms were proposed for triangle counting, such as the work in~\cite{becchetti2008efficient}, however, they are not practical and produce large error as discussed in~\cite{pavan2013counting}.

Horvitz-Thompson estimation was proposed in the graphical setting by Frank~\cite{frank1978samplingsocialnetworks}, including applications to subgraph sampling, but limited to a model of simple random sampling of vertices without replacement; see also  Kolaczyk~\cite{kolaczyk2009statistical}.

\paragraph{\textbf{\textit{Graph Analysis Using Static and Parallel Algorithms}}}

We briefly review other research for graph analysis in non-streaming setting (i.e., static). For example, exact counting of triangles with runtime ($O(m^{3/2}$)~\cite{schank2007algorithmic}, or approximately by sampling edges as in~\cite{tsourakakis2009doulion}. Although not working in a streaming fashion, the algorithm in~\cite{tsourakakis2009doulion} uses unbiased estimators of triangle counts similar to our work. Moreover, other algorithms were proposed based on wedge sampling and proved to be accurate in practice, such as the work in~\cite{seshadhri2013fast,seshadhri2013wedge,kolda2013counting}. More recently, the work done in~\cite{rossi2014pmc} proposed a parallel framework for finding the maximum clique.

Finally, there has been an increasing interest in the general problem of network sampling. For example, to obtain a representative subgraph~\cite{leskovec2006slg,ahmed2013network}, to preserve the community structure~\cite{Maiya2010www,Maiya2011kdd}, to perform A/B testing of social features~\cite{backstrom2011network}, and many other interesting work~\cite{dasgupta2012social,al2009output,vattani2011preserving}.

\section{Conclusion}
\label{sec-conclusion}
In this paper, we presented a generic framework for big-graph analytics called graph sample and hold (gSH). The gSH framework samples from massive graphs \emph{sequentially in a single pass}, one edge at a time, while maintaining a small state typically less than $1$\% of the total number of edges in the graph. Our contributions can be summarized in the following points:
\begin{itemize}
\item gSH works sequentially in a single pass, while maintaining a small state.
\vspace{-2.mm}
\item We show how to produce unbiased estimators and their variance for four specific graph quantities of interest to estimate within the framework. Further, we show how to obtain confidence bounds using the variance unbiased estimators.
\vspace{-2.mm}
\item We conducted several experiments on real world graphs, such as social Facebook graphs, and web graphs. The results show that the relative error goes from $0.02$\% to $0.95$\% for a sample with $\leq 40K$ edges, across different types of graphs. Moreover, the results show that the sampling distribution is centered and balanced over the actual values of the four graph quantities of interest, with tightening error bounds as the sample size increases.
\vspace{-2.mm}
\item We discuss the effect of parameter choice $p,q$ on the proportion of sampled edges.
\vspace{-2.mm}
\item We compare to the state of the art~\cite{jha2013space}, and our proposed framework has a relative error \emph{orders of magnitude} less than the Streaming-Triangles algorithm proposed in ~\cite{jha2013space}, as well as with a small(er) overhead storage (in most of the graphs). We note that the work in~\cite{jha2013space} compares to other state of the art algorithms and shows that they are not practical and produce a very large error; see Section~\ref{sec-related} for more details.
\vspace{-2.mm}
\item We show how to parallelize and efficiently compute the unbiased variance estimators, and we discuss the significant reductions in computation time that can be achieved by gSH framework.    
\vspace{-2.mm}
\end{itemize}  

In future work, we aim to extend gSH to other graph properties, such as cliques, coloring number, and size of connected components, among many others. 


\bibliographystyle{acm}
{
\bibliography{paper}
}

\end{document}